\documentclass[10pt,conference]{article}
\usepackage{adjustbox}
\usepackage{algorithmicx}
\usepackage{amsmath}
\usepackage{amsfonts}
\usepackage{amssymb}
\usepackage{amsthm}
\usepackage{authblk}
\usepackage{comment}
\usepackage{color}
\usepackage[capitalise]{cleveref}
\usepackage{enumitem,kantlipsum}
\usepackage{stmaryrd}
\usepackage{xspace}

\newtheorem{lemma}{Lemma}
\newtheorem{definition}[lemma]{Definition}
\newtheorem{example}[lemma]{Example}
\newtheorem{exercise}[lemma]{Exercise}

\newtheorem{theorem}{Theorem}
\newtheorem{corollary}[lemma]{Corollary}
\newtheorem{observation}[lemma]{Observation}

\newtheorem{conjecture}[theorem]{Conjecture}

\newcommand{\UC}{\mathbb{C}_\FactSet}

\newcommand{\Islandh}{h_{M_\FactSetAux}}
\newcommand{\Islands}[0]{\mathcal{I_\FactSet}}
\newcommand{\Globalh}{\bar{h}_\FactSet}

\newcommand{\Tuple}[1]{\ensuremath{\langle #1 \rangle}\xspace}

\newcommand{\FormatEntitySet}[1]{\ensuremath{\mathbf{#1}}\xspace}
\newcommand{\FormatFormulaSet}[1]{\ensuremath{\mathcal{#1}}\xspace}

\newcommand{\mappings}[1]{\mathcal{H}om({#1})}

\newcommand{\Predicates}{\FormatEntitySet{\Sigma}}

\newcommand{\Rule}{\ensuremath{\rho}\xspace}
\newcommand{\BCQ}{\ensuremath{\phi(\bar y)}\xspace}

\newcommand{\Body}{\ensuremath{\beta}\xspace}

\newcommand{\FactSet}{\ensuremath{\mathbb{D}}\xspace}
\newcommand{\FactSeta}{\ensuremath{\mathbb{F}}\xspace}
\newcommand{\factseta}{\ensuremath{\mathbb{F}}\xspace}
\newcommand{\factset}{\FactSet}
\newcommand{\FactSetAux}{\factseta}

\newcommand{\RuleSet}{\ensuremath{\FormatFormulaSet{T}}\xspace}
\newcommand{\ruleset}{\RuleSet}
\newcommand{\NormRuleSet}{\FormatFormulaSet{T}_{NF}}
\newcommand{\DLRuleSet}{\FormatFormulaSet{T}_{DL}}
\newcommand{\EXRuleSet}{\FormatFormulaSet{T}_{\exists}}

\newcommand{\Step}[2]{\ensuremath{\textit{Ch}_{#1}(#2)}\xspace}
\newcommand{\Chase}[1]{\ensuremath{\textit{Ch}(#1)}\xspace}
\newcommand{\sh}[1]{\textit{sh}(#1)}
\newcommand{\rsize}{r\hspace{-0.5mm}s}
\newcommand{\aaaddd}[1]{\bar a\in dom(#1)^{|\bar y|}}
\newcommand{\aaaddda}[1]{\bar a\in dom(#1)^{|\bar a|}}
\newcommand{\nic}[1]{\text{\tiny #1}}

\newcommand{\ChaseEX}[1]{\ensuremath{\textit{Ch}_{\exists}(#1)}\xspace}
\newcommand{\ChaseEXi}[2]{\ensuremath{\textit{Ch}_{#2, \exists}(#1)}\xspace}

\newcommand{\Core}[1]{Core(#1)}

\newcommand{\Schema}{\Sigma}

\newcommand{\Skeleton}{S}

\newcommand{\cparent}{cpar}

\newcommand{\cancestor}{canc}
\newcommand{\set}[1]{\{#1\}}
\newcommand{\pair}[1]{\langle#1\rangle}

\newcommand{\RuleSetD}{\ruleset_d}
\newcommand{\rulesetd}{\ruleset_d}

\newcommand{\Signature}{\Schema}
\newcommand{\ParentFunction}{par}
\newcommand{\AncestorFunction}{anc}

\newcommand{\tgdl}[2]{\makebox[0.5cm][r]{}\makebox[1.8cm][l]{#1} \makebox[2cm]{$\rightarrow$} \makebox[2cm][l]{#2}}

\newcommand{\M}{Q}
\newcommand{\setM}{\mathcal{Q}}

\newcommand{\hypenoperator}[2]{\operatorname{#1-#2}}

\newcommand{\cutred}{\hypenoperator{cut}{red}}
\newcommand{\cutgreen}{\hypenoperator{cut}{green}}
\newcommand{\mergered}{\hypenoperator{fuse}{red}}
\newcommand{\mergegreen}{\hypenoperator{fuse}{green}}

\newcommand{\reduce}{\operatorname{reduce}}

\newcommand{\GmJ}{G^{\text{\tiny -1}}}
\newcommand{\RmJ}{R^{\text{\tiny -1}}}

\begin{document}
\title{A Journey to the Frontiers of Query Rewritability}
%\author{}
%{Piotr Ostropolski-Nalewaja, Jerzy Marcinkowski, David Carral and Sebastian Rudolph\\
%Institute of Computer Science, University of Wrocław\\
%TU Dresden, Computational Logic Group, Germany}
% 
 \author[1]{Piotr Ostropolski-Nalewaja}
 \author[1]{Jerzy Marcinkowski}
 \author[2]{David Carral}
 \author[3]{Sebastian Rudolph}
 \affil[1]{Institute of Computer Science, University of Wrocław}
 \affil[2]{TU Dresden,  Knowledge-Based Systems Group, Germany}
 \affil[3]{TU Dresden, Computational Logic Group, Germany}

\maketitle
\begin{abstract}
This paper is about (first order) query rewritability in the context of theory-mediated query answering. The starting point of our journey is the FUS/FES conjecture, saying that if a theory is core-terminating (FES) and admits query rewriting (BDD, FUS)
then it is uniformly bounded. We show that this conjecture is true for a wide class of ``local'' BDD theories. Then we 
ask how non-local can a BDD theory actually be and we discover phenomena which we think are quite counter-intuitive.

\end{abstract}

\section{Introduction}

The scenario we consider in this paper  has been studied extensively both in the database theory and in the context of description logics: there 
is a database {\em instance} \factset (also called the {\em fact set}, or {\em structure}) and the {\em  theory} \ruleset (or {\em the rule set}), which is a set of {\em tuple generating dependencies}  (or {\em rules}).
For a given  conjunctive query $\phi$ we want to know\footnote{This is sometimes called {\em ontology mediated query answering.}} whether \factset and \ruleset jointly entail (logically imply) $\phi$ (denoted as
$\factset, \ruleset\models \phi$).\\

\noindent
{\bf  Chase and Core Termination.} The notion of {\em Chase} is fundamental in this context. $Ch(\factset,\ruleset)$ is a  structure  constructed, from 
\factset, by the {\em chase procedure}, that is by adding, to the current structure, new terms and atoms, witnessing that 
the constraints from \ruleset are satisfied, and doing it until a fixpoint is reached. 

The structure $Ch(\factset,\ruleset)$ (which is also called {\em Chase}) is constructed in steps: $Ch_0(\factset,\ruleset)$ is defined as \factset. Then, 
for given  $Ch_i(\factset,\ruleset)$ the structure  $Ch_{i+1}(\factset,\ruleset)$ is constructed by adding (in parallel) to 
 $Ch_i(\factset,\ruleset)$ all atoms\footnote{This may involve creating new elements of the structure, called in database theory {\em terms}.} required by rules from \ruleset for elements of  $Ch_i(\factset,\ruleset)$. 
 Then, obviously,  $Ch(\factset,\ruleset)$ is defined as $\bigcup_{i\in \mathbb N} Ch_i(\factset,\ruleset)$. 
 
 It is well known \cite{chaserevisited} that $\factset, \ruleset\models \phi$ if and only if $\phi$ is true in $Ch(\factset,\ruleset)$:
 {
 %\small 
\[ \forall \ruleset \; \forall \factset \; \forall \phi \;\;\; ( Ch(\factset,\ruleset)\models \phi \;  \Leftrightarrow \;  \factset,\ruleset \models \phi  \ ) \] }
We say that \ruleset is {\em Core Terminating} (or that it has the FES property, where FES stands for Finite Expansion Set \cite{fusfesdef}) if, regardless of \factset,
all conjunctive queries satisfied in $Ch(\factset,\ruleset)$ are already satisfied in some  $Ch_i(\factset,\ruleset)$. More precisely, \ruleset is Core Terminating if:
{
%\small 
\[ \forall \factset \; \exists i\in {\mathbb N} \; \forall \phi \;\;( Ch(\factset,\ruleset)\models \phi \;  \Leftrightarrow \;  Ch_i(\factset,\ruleset) \models \phi ) \;
\tag{\text{\footnotesize CT}  }\]}
This is an important property, since $Ch(\factset,\ruleset)$ is typically an infinite structure, only existing as an  abstract mathematical object, and impossible to query,  while $Ch_i(\factset,\ruleset)$ is  always finite and so in principle it can be constructed and queried.\\

 \noindent
{\bf The Bounded Derivation Depth property.} But {\bf the} most important property a theory  can possibly enjoy in this context is the {\em Bounded Derivation Depth Property} (or {\em Finite Unification Set Property}, FUS). We say that \ruleset has the { Bounded Derivation Depth} property (or ``is BDD'') if conjunctive queries always rewrite: for each $\phi$ there exists a query $\phi_\ruleset$, being a union of conjunctive
queries\footnote{This is known to be equivalent to the existence of first order rewriting.}, such that for every \factset we have that 
 $Ch(\factset,\ruleset)\models \phi $ if and only if $\factset\models \phi_\ruleset$. See how extremely  useful it is: instead of querying  $Ch(\factset,\ruleset)$, 
 an elusive infinite structure, we can equivalently  query \factset, the only structure we really have access to. No wonder that BDD/FUS property has been considered in literally hundreds of papers (or maybe more than that). 
 
 Important classes of BDD theories have been identified, and intensively studied, among them  decidable subclasses like:\\
 \textbullet~ linear theories, where rules only can have at most one atom in the body;\\
 \textbullet~ guarded BDD theories -- while not all guarded theories are BDD, as \cite{BBLP18} and \cite{CR15} show it is decidable whether a guarded theory is BDD, 
 so guarded BDD theories are a decidable subclass, generalizing linear theories;\\
 \textbullet~ sticky theories, building on the idea that it is  unrestricted join that makes theories non-BDD, and defined by a 
 reasonably natural syntactic restriction on the use of join \cite{CGP10}.
 
 Apart from the decidable subclasses of BDD there are some well known and natural undecidable subclasses:\\
 \textbullet~ bounded Datalog theories, already studied decades before the class BDD itself  has been discovered \cite{GMSV87};\\
 \textbullet~ binary BDD theories, where the arity of relation symbols is at most 2, studied in the context of description logics.

Another interesting, but less well known (undecidable) subclass of BDD (and superclass of sticky theories) are:\\
 \textbullet~ backward shy theories, from 
\cite{T13}.

As we explain in this paper, despite all this effort, we only understand very little about the deeper mathematical properties of BDD theories. 
 In particular we show that the intuitive understanding of BDD as theories which are ``local'', ``only depending on the small pieces of \factset''
 and ``unable to look too far''  is (more or less) correct for all aforementioned classes of  BDD theories but incorrect for the BDD theories in general.\\

 \noindent
 {\bf The FUS/FES conjecture.} As it turns out\footnote{See \cref{thm:bdd-th}.} (and is not terribly hard to prove \cite{CGL09}) theory \ruleset is BDD if and only if\footnote{In Section \ref{sec-bdd} we use this characterization as the definition of the class BDD.}:
 {
 \small
 \[ \forall \phi\; \exists i\in{\mathbb N}\; \forall \factset \;\;(Ch(\factset,\ruleset)\models \phi \; \Leftrightarrow \;  Ch_i(\factset,\ruleset) \models \phi )\;  \tag{\text{\footnotesize BDD}  }\]}
 Which means that, in order to evaluate $\phi$ is is always enough to run only first $i$ steps of the chase, with $i$  depending on $\phi$ but not on $\factset$. This is useful. But would be 
 even  more useful if $i$ could be chosen in a uniform way, independently of $\phi$.
 
 There is a striking similarity between the formula (BDD) and (CT). And a natural conjecture,
 which we call the FUS/FES conjecture,
 is that for a BDD theory which is, at the same time,
 Core Terminating, the number $i$ can be indeed chosen in a {\bf uniform way}. In other words the conjecture, 
 says that if \ruleset is both BDD and Core Terminating then:
 {
 %\small
 \[ \exists i\in{\mathbb N}\; \forall \phi \; \forall \factset \; Ch(\factset,\ruleset)\models \phi \; \Leftrightarrow \; Ch_i(\factset,\ruleset) \models \phi \;  \tag{\text{\footnotesize UBDD}  }\]}
 This conjecture was earlier studied in  \cite{LMU16} where a proof was proposed, which is however incorrect and was later withdrawn, and in \cite{BLMTU19} where it is proved that it would hold true if the assumption that \ruleset is Core-Terminating was replaced by a (much stronger) assumption that it is All-Instances Terminating. \\

 \noindent
 {\bf This paper.} Our main technical result  is that the FUS/FES conjecture holds for a wide class of {\em local} (as we call them) BDD theories. 
 This class includes most of the aforementioned subclasses of BDD (apart from sticky and backward shy), in particular all BDD theories over binary signatures\footnote{This is how this paper originated: we attempted to attack the FUS/FES conjecture, we managed to prove it for the binary case, and we tried to understand how far our proof can generalize, and why it does not seem to generalize to the entire BDD class.}.

 But what we find at least as interesting is the discovery of theories which are very much non-local (and not even what we call {\em bounded-degree local}). The existence of such theories not only defies all the intuitions that have been built 
 about the BDD class but also shows that so far we have probably been able to barely scratch the surface of the BDD class, and that there is  a lot of room for new decidable/syntactic classes of BDD theories, richer than all that was considered so far\footnote{But, as we also explain in this paper, the unexpected and 
 counter-intuitive phenomena only may appear for theories of arity higher than 2. This is probably exactly why they are counter-intuitive: theories of arity 2 are much easier to imagine and they are mainly responsible for shaping our intuitions. }.
 
 To be more precise, in this paper:
 
 \noindent
 (1) We define the notion of {\em local} theories and show that each local theory is BDD. We also notice that linear theories are local (this is all easy, Section \ref{local-theories})  and that,
 if \ruleset is local then the size of each disjuncts of the rewriting $\phi_\ruleset$ of a CQ $\phi$ is linear in the size of $\phi$.
 
  \noindent
 (2) We prove that the FUS/FES conjecture holds for local theories -- every local theory that satisfies (CT) also satisfies (UBDD) (this is not that easy,  Section \ref{conjecture-for-local}). 
 
 \noindent
 (3) We show that all BDD theories over binary signatures are local, and remark  that this also holds true for all guarded BDD theories (not totally easy,  Sections \ref{local-theories} and \ref{proof-for-binary}). This means that the FUS/FES conjecture holds for such theories.
 
 \noindent
 (4) We notice that if infinite theories are allowed then the FUS/FES conjecture fails, even for binary signatures (very easy,  Section \ref{fus-fes-conjecture}).
 
 \noindent
 (5) We notice (Section \ref{slightly-beyond}) that sticky theories, while BDD, are not always local. We define another, weaker, notion, of {\em bounded degree local} theories (or {\em bd-local}). We notice that sticky theories 
 are always bd-local\footnote{For arity above 2 we only consider connected theories, whose all rules have connected bodies.}, so that all the 
 theories from previously studied decidable BDD classes are (at least) bd-local. We also notice that 
 if  attention is restricted to database instances  of bounded degree  then the two notions of locality coincide (and,
 in such case, the FUS/FES conjecture also holds for sticky theories).
 
 \noindent
 (6) We show a BDD theory which is not bounded degree local (easy, once you know it, Section \ref{slightly-beyond}).
 
 \noindent (7) In Section \ref{far-beyond} we examine the intuition that ``BDD theories are unable to look too far''. We define the notion of {\em   distancing  } theories, and 
 show that if a theory is  local then it is also distancing (easy). We also notice that  backward shy theories are obviously distancing, so that all the previously known examples of BDD theories are indeed   distancing. We show however, that there exists a BDD theory $\rulesetd$ which is not distancing (this is quite complicated, Sections \ref{far-beyond}-\ref{far-beyond-dwa} and \ref{app-2}).
 As a corollary we get that, for this BDD theory, the rewriting $\phi_{\rulesetd}$ (which is a disjunction of conjunctive queries) of a query $\phi$ can require 
 disjuncts of exponential size with respect to the size of $\phi$. In view of (1) this is in stark contrast to the previously known BDD theories\footnote{There is a bit of confusion here. A folklore belief is that existence of theories which require 
 unbounded size of rewritings is a consequence of the fact that BDD is an undecidable property of the theory (see for example the 
 stackexchange post \cite{K17}). But we never saw a detailed argument based on this proof idea. And this is probably due to the fact 
 that such argument does not exist. Why? Because being BDD is already an undecidable property for theories with binary signature, and still 
 such theories, if BDD, are local and thus admit rewritings of linear disjunct size.
 
 What we write may also sound confusing for readers who know \cite{rewriting-size-A} (and also \cite{rewriting-size-B}) where exponential lower bound is shown for 
 the size of rewritings. But \ruleset is a part of the input in both \cite{rewriting-size-A} and \cite{rewriting-size-B}, while for us  the size of rewritings 
 with respect to (a fixed) \ruleset is a  measure of the complexity of \ruleset.}
 (since also all backward shy theories admit rewritings of linear disjunct size).
 
 \noindent (8) Then, in Section \ref{ojojoj},  we go even one (huge) step further, showing (by a generalization of the techniques developed in Sections \ref{far-beyond}-\ref{far-beyond-dwa})
 that for each$K\in \mathbb N$ there exists a BDD theory $\rulesetd^K$ and a CQ $\phi(\bar y)$, such that the rewriting $\phi_{\rulesetd^K}$ has
 disjuncts of size $K$-fold exponential in the size of $\phi$.

 The rest of the paper is organized as follows. In Section \ref{section:preliminaries} we  recall the standard database theory notions which are used in this paper. Nothing surprising  happens there. Then,
 in Sections \ref{sec-skolem}--\ref{sec-core-termination} we recall definitions of semi-oblivious chase,
 and of the fundamental properties of theories: BDD and Core-Termination property. Our presentation is sometimes slightly non-standard, and tailored for the needs of the sections that follow, so it maybe would not be the best idea to skip them. Then, in Sections \ref{fus-fes-conjecture}--\ref{far-beyond-dwa} 
 our new findings are presented (with some proofs deferred to  the Appendix).
 
 %%%%%%%%%%%%%%%%%%%%%%%% PRELIMINARIES
 
 \section{Preliminaries} 
\label{section:preliminaries}

\noindent
{\bf Queries and TGDs.}
A {\em conjunctive query} (CQ) is a  formula $\psi(\bar y) = \exists \bar x \; \Body(\bar x, \bar y)$ with \Body being a non-empty conjunction of atomic formulas over some {\em signature} (or {\em schema}) \Predicates (which is a finite set of relation symbols)
and over some set of  variables and set of constants. So, for example 
$\exists x\; Siblings(Abel,x), Female(x)$ is\footnote{Assuming that $Siblings,Female\in \Predicates$ and  $Abel$ is a constant.} a CQ.

A CQ is \emph{Boolean} (or BCQ) if all variables are quantified (as in the above example). We refer to \Body as the \emph{body} of  $\psi(\bar y)$.
By a \emph{union of conjunctive queries} (UCQ) we mean a formula being a disjunction of CQs. By the \emph{size} of a CQ, denoted $|\psi(\bar y)|$, we  mean the number of atomic formulas it is built of. 

\noindent
A {\em theory} or a \emph{rule set} is a finite\footnote{With one exception in Example \ref{infinite} where an infinite theory is considered, with an infinite signature.} set of \emph{Tuple Generating Dependencies}   (or \emph{rules}).  A  TGDs is a first-order logic  formula of the form: \vspace{-3.5mm}
$$
\forall \bar x, \bar y \; \big(\Body(\bar x, \bar y) \Rightarrow  \exists \bar w \; \alpha(\bar y, \bar w)\big)\vspace{-1mm}
$$
where $\bar x$, $\bar y$  and $\bar w$ are pairwise disjoint lists of variables, $\Body(\bar x, \bar y$) (the rule's \emph{body}) is a conjunction of atomic formulas and 
$ \alpha(\bar y, \bar w)$ (the rule's \emph{head}) is an atomic formula\footnote{Using database theory terminology, our rules are ``single head'' Tuple Generating Dependencies. This is the only reasonable choice in this context, since we want to talk about theories over binary signature: if we allowed multi-head rules, with heads comprising several atoms, then rules with predicates of any arity could be easily simulated using only arity 2 predicates.}. 
The \emph{frontier} $\bar y$ of a rule, denoted $\mathit{fr}(\rho)$, is the set of all variables that occur both in the body and the head of the rule.
We omit universal quantifiers when writing rules and treat conjunctions of atoms, such as \Body above, as atom sets.

\noindent
{\bf Structures and entailment.}
A database {\em instance}  
(or {\em structure}\footnote{One usually thinks that instances and fact sets are finite, 
while structures can also be infinite.} or a {\em fact set}) is a set of 
 facts -- atomic formulas over \Predicates . For a structure $\FactSeta$ over \Predicates by $dom(\FactSeta)$ we denote its {\em active domain} -- the set of all 
terms which appear in the facts of  $\FactSeta$. For $c,c'\in dom(\FactSeta)$ by $dist_\FactSeta(c,c')$ we mean the 
distance between $c$ and $c'$ in the Gaifman graph of $\FactSeta$: the vertices of this graph are elements of $dom(\FactSeta)$ and 
two vertices are connected with an edge if and only if they appear in the same fact.

  We  write $\factset\models \ruleset$ to say that $\factset$ is a model of $\ruleset$ -- all the TGDs from $\ruleset$ are satisfied in $\factset$.
For a pair $\RuleSet, \FactSet$, a CQ $\phi(\bar y)$, and a tuple $\aaaddd{\factset}$  we write $\RuleSet, \FactSet\models \phi(\bar a)$ to indicate that \Tuple{\RuleSet, \FactSet} entails 
$\phi(\bar a)$,  which means that $\phi(\bar a)$ is true in each structure $\factseta$ such that $\factseta\models \ruleset$ and 
$\FactSet\subseteq \factseta$.

\begin{example}\label{abel}
Imagine theory $\ruleset_a$ consisting of two rules:

\noindent
\hspace*{2mm} $Human(y)\Rightarrow \exists z\; Mother(y,z)$\\
\hspace*{2mm} $Mother(x,y) \Rightarrow Human(y)$

\noindent
and  an instance $\factset_a =\{Human(Abel)\}$. Then:\\ 
\hspace*{2mm}$\RuleSet_a, \FactSet_a \models \exists y,z\; Mother(Abel,y), Mother(y,z)$.
\end{example}

\noindent
{\bf Homomorphisms and Query Containment.}
For two structures $\factset$ and $\factseta$, a homomorphism from $\factset$ to $\factseta$ is a function
$h: dom(\factset) \to dom(\factseta)$ such that for each $\alpha\in \factset$ there is $h(\alpha)\in \factseta$.
Notice that using this notation we slightly abuse types since, formally speaking, $\alpha\not\in Dom(h)$. 

% We will use the following easy facts:
% 
% \begin{observation}\label{homo-properties}
% \textbullet~
% If  $\factseta$ is the image of $\factset$ via 
% %some
% homomorphism $h$ (that is $\factseta = \{h(\alpha): \alpha\in \factset\}$) and  $\factset \models \ruleset$  then $\factseta \models \ruleset$.\\
% \textbullet~
% Composition of homomorphisms is a homomorphism.

% \end{observation}

For two CQs $\phi(\bar y)$ and  $\psi(\bar y)$, with the same set of free variables, we say that $\phi(\bar y)$ contains $\psi(\bar y)$
if for every structure $\factset$ and for every tuple $\aaaddd{\factset}$ if $\factset\models \phi(\bar a)$ than 
also $\factset\models \psi(\bar a)$. It is well known that $\phi(\bar y)$ contains $\psi(\bar y)$ if and only if
there is a homomorphism\footnote{Queries $\phi(\bar y)$ and $\psi(\bar y)$
are seen as structures here: active domains of these structures are sets of the variables of  $\phi(\bar y)$ and $\psi(\bar y)$. }
from $\phi(\bar y)$ to $\psi(\bar y)$ which is the identity on variables from $\bar y$.

\noindent
{\bf Connected queries, rules and theories. }
For CQ one can in a natural way define its Gaifman graph. Variables are the vertices of this graph and two variables are connected by 
an edge if and only if they both appear in the same atomic formula. A conjunctive query is \emph{connected} if its Gaifman graph
is connected. A TGD is \emph{connected} if its body is. A theory is \emph{connected} if each of its rules is.

{\bf \small All the theories we consider in this paper are connected} (with the important {\bf \small exception}  for theories 
over a binary signature). This does not reduce the expressive power of such theories due to the following trivial trick: add a fresh variable
as an additional, first variable in all the  atoms appearing in the rules of the theory. Not only this 
will make the theory connected,
but it will obviously preserve its BDD and Core Termination status. But it will increase the arity -- so if we care about the arity we do not 
get connectivity for free. 

Notice that after  applying  the trivial trick to an instance $\FactSeta$  the distance between each $c$ and $c'$ from 
$dom(\FactSeta)$ will be at most 2. Also, applying this trick turns any instance with  Gaifman graph of a low (bounded) degree into one 
with high degree Gaifman graph.

\section{The Skolem Chase.}\label{sec-skolem}

The Chase is a standard algorithm, studied in a plethora of papers. It  semi-decides whether 
$\ruleset,\factset \models \phi(\bar a)$ for given theory \ruleset, instance \factset, CQ $\phi(\bar y)$ and tuple $\aaaddd{\factset}$. The Chase comes in many variants and flavors.  The best way to present our results is by using the {\bf semi-oblivious Skolem chase}\footnote{Or, to be more precise, ``Semi-oblivious chase with the Skolem naming convention''.}, 
which we define in this Section. 

%%%%%%%%%%%%5

\begin{definition}\label{ftau}
For a CQ $\phi(\bar y)$, with  body  consisting of a single atom, possibly preceded by some existential quantifiers,  define $\tau(\BCQ)$ as the isomorphism type of $\BCQ$. 
\end{definition}

We mean here that  $\tau(\BCQ)$ depends on the relation symbol of $\phi(\bar y)$, on the equalities between its variables and on the set of positions in $\phi(\bar y)$ which are occupied by variables\footnote{The type $\tau(\BCQ)$ should also depend on the constants occurring in $\BCQ$  but we will only apply this notion to queries without constants.} quantified  in $\BCQ$, but not on the names of the variables. 

% 
% Let $\tau = \tau(\phi)$ be the isomorphism type of some conjunctive query $\phi = \exists{\bar x} \; P(\bar x, \bar y )$ (as in \cref{ftau}). For each $i \in [1,arity(P)]$ we define a new function symbol $f_i^\tau$ of arity $|\bar y |$.

For each possible isomorphism type $\tau=\tau(\phi(\bar y))$, for some conjunctive query $\phi(\bar y)$, as in Definition \ref{ftau}, and for each 
natural number $1\leq i\leq arity(E)$, where $E$ is the relation of 
$\phi(\bar y)$, let $f_i^\tau$ be a function symbol,
with arity equal to $|\bar y|$, that is the number of free variables in $\phi(\bar y)$. 

\begin{definition}[Skolemization]\label{skolemization}
For given TGD \Rule, of the form 
$\Body(\bar x, \bar y ) \Rightarrow \exists \bar w \; \alpha(\bar x, \bar w)$
by $\sh{\Rule}$ we denote the \emph{skolemization of the head of } \Rule, which is 
the atom $\alpha(\bar x, \bar w)$, with each variable $w\in \bar w$
replaced by the  term $f_i^\tau(\bar x)$, where $i$ is the earliest 
position in $\alpha(\bar x, \bar w)$ where the variable $w$ occurs.
\end{definition}

\noindent
Let, for example $\Rule$ be $E(x,y,z), P(x)\Rightarrow \exists v\; R(y,v,z,v)$. Then 
$\sh{\Rule}$ will be the atom $R(y,f_2^\tau(y,z),z,f_2^\tau(y,z)) $ where $\tau$ is the isomorphism type of
%the query 
$\exists v\; R(y,v,z,v)$. Notice that $\sh{\Rule}$ does not depend on the body of $\Rule$, only on its head. In particular
it does not depend on the non-frontier variables\footnote{Including non-frontier variables as arguments of the functions
$f_i^\tau$ (like $x$ in the current example) would lead to Oblivious chase.} of the body of $\Rule$. 

Now we can define the procedure of Rule Application. Parameters of this procedure are an instance $\FactSeta$, a rule $\Rule$, and a mapping 
$\sigma$ assigning elements of the active domain of $\FactSet$ to the variables which occur in the body of $\Rule$:

\begin{definition}[Rule Application]\label{application}
Let \Rule be a rule of the form 
$\Body(\bar x, \bar y ) \Rightarrow \exists \bar w \; \alpha(\bar y, \bar w)$,
and let \FactSeta be a fact set.

\begin{itemize}
\item
 Define $\mathcal{H}om(\Rule,\FactSeta)$
as the set of all
mappings $\sigma$ from variables in $\bar x \cup \bar y $ to $dom(\factseta)$
such that 
$\sigma(\Body(\bar x, \bar y ))\subseteq\FactSeta$ (which means that all the atoms from \Body are in \FactSeta after we apply $\sigma$ to them).

\item 
For $\sigma\in \mathcal{H}om(\Rule, \FactSeta)$ define 
 $appl(\Rule, \sigma)= \sigma(\sh{\Rule})$.
 
 \end{itemize}
\end{definition}

Which, in human language, means that  $appl(\Rule, \sigma)$ is the atom that we expect to see in all structures satisfying \Rule to which the 
body of \Rule can be mapped via a homomorphism $\sigma$.

The chase procedure can now be defined. It produces,  for given  instance $\FactSet$ and theory $\RuleSet$, a sequence $\{\Step{i}{\RuleSet, \FactSet}\}_{i\in \mathbb N} $ of instances and the structure \Chase{\RuleSet, \FactSet}:

\begin{definition}[Semi-Oblivious Skolem chase procedure]\label{skolem-chase}
~\\
\textbullet~ Define  $\Step{0}{\RuleSet, \FactSet}=\FactSet$.\vspace{1.5mm} \\
\noindent
\textbullet~ Assume $\Step{i}{\RuleSet, \FactSet}$ is defined. Then:~
$\Step{i+1}{\RuleSet, \FactSet}=$ \vspace{1mm} \\
\mbox{\hspace{-1mm}$\Step{i}{\RuleSet, \FactSet}\cup \{ appl(\Rule, \sigma): \Rule\in \ruleset, \sigma\in \mathcal{H}om(\Rule,\Step{i}{\RuleSet, \FactSet})  \}.$}\vspace{1.5mm}\\
% -- {\bf for each} rule $\Rule$ from $\RuleSet$, of the form like in Definition \ref{application}
% \hspace*{6.5mm}and for each substitution $\sigma$ {\bf do in parralel:}\\
% \hspace*{10mm} run Rule Application for rule $\Rule$, substitution $\sigma$,\\ 
% \hspace*{10mm} for $\Step{i+1}{\RuleSet, \FactSet}$ as the instance \FactSeta;\\
% -- {\bf return } $\Step{i+1}{\RuleSet, \FactSet}$.
\noindent
\textbullet~  Define $\Chase{\RuleSet, \FactSet}= \bigcup_{i\in \mathbb N} \Step{i}{\RuleSet, \FactSet}$.
\end{definition}

It is well known, that \Chase{\RuleSet, \FactSet} is 
 a universal model for $\RuleSet$ and $\FactSet$ (i.e., a model that can be homomorphically mapped into any other model).
Therefore, this structure can be  used to solve CQ entailment: for any theory $\ruleset$, CQ $\phi(\bar y)$, instance $\factset$  and $\aaaddd{\factset}$ it holds that:
 \vspace{-1mm}
$$ Ch(\ruleset, \factset)\models \phi(\bar a) \;  \Leftrightarrow \;  \factset,\ruleset \models \phi(\bar a)$$   
\begin{example} Let $\ruleset_a$ and $\factset_a$ be as in Example 
\ref{abel}. 
Then we have  $\Step{0}{\RuleSet_a, \FactSet_a}=\{Human(Abel)\}$ and\\
$\Step{1}{\RuleSet_a, \FactSet_a}=\{Mother(Abel, mum(Abel))\}\cup$ $\Step{0}{\RuleSet_a, \FactSet_a}$. Then 
$\Step{2}{\RuleSet_a, \FactSet_a}=\Step{1}{\RuleSet_a, \FactSet_a}\cup$  $\{ Mother(mum(Abel), mum(mum(Abel)))\}$.  And so on
(we use the function symbol ``mum''  as an alias for
the ugly Skolem function symbol from Definition \ref{skolemization}). 
\end{example}

Now notice that there is nothing in Definition \ref{skolem-chase} that could prevent us from taking $\factset=\Step{2}{\RuleSet_a, \FactSet_a}$ and running 
chase for such $\factset$. It is easy to see that in such case we get $\Chase{\RuleSet_a, \FactSet}=\Chase{\RuleSet_a, \FactSet_a}$.
This leads to an easy:

\begin{observation}\label{posrodku}
If $\factset \subseteq \FactSeta \subseteq \Chase{\RuleSet, \FactSet}$ then $\Chase{\RuleSet, \FactSeta}=\Chase{\RuleSet, \FactSet}$.
\end{observation}

It is important for us, and constitutes the main reason why we decided to use the Skolem naming convention, that the equality in  Observation \ref{posrodku}
is to be understood literally (rather than  ``up to isomorphism'').

%%%%%% dodany chaos %%%%%%%%%%%%%%%%%%%%%%%%%%%%%%%%%

\vspace{1mm}
\noindent
{\bf The frontier and the birth atoms.}
Let $\alpha$ be an atom from $ Ch(\ruleset, \factset)\setminus \factset$, created as $appl(\rho, \sigma)$ for some $\rho\in \ruleset$ 
and $\sigma\in {H}om(\Rule,\FactSet)$.
By frontier of $\alpha$ (denoted $\mathit{fr}(\alpha)$) we will mean the set of terms $\sigma(\mathit{fr}(\rho))$. Notice that there may be more than one 
rule application creating  the same atom $\alpha$, but:

\begin{observation}
If $appl(\rho, \sigma)=appl(\rho',\sigma')$ then $\mathit{fr}(\rho)=\mathit{fr}(\rho')$ and $\sigma$ and $\sigma'$ agree on $\mathit{fr}(\rho)$.
\end{observation}

Clearly, for each $t\in dom(Ch(\ruleset, \factset))$ either there is $t\in dom(\factset)$ or $t$ was created by chase as a Skolem term:

\begin{observation}\label{lem:termparentatom}
Suppose $t\in dom(Ch(\ruleset, \factset))\setminus dom(\factset)$. Then there exists exactly one atom $\alpha\in Ch(\ruleset, \factset)$ such that 
$t$ appears in $\alpha$, but $t\not\in \mathit{fr}(\alpha)$.
\end{observation}

Such atom will be called the {\em birth atom} of $t$.

%%

 %%%%%%%%%%%%%%%%%%%%%%%%%%%%%%%%% BDD DEFINITIONS AND EXERCISES
 
 \section{The Bounded Derivation Depth Property}\label{sec-bdd}

For an instance \factset, a natural number $n$ and a query $\phi(y)$ we will write\footnote{``it is {\em enough} to run $n$ steps of \ruleset-chase
on \factset to 
see if $\phi$ holds'' } $Enough(n,\phi(y),\factset,\ruleset)$ as a shorthand for:
 \vspace{-1mm}
$$ \forall \aaaddd{\factset}\;\; (Ch(\ruleset, \factset)\models \phi(\bar a)\Leftrightarrow Ch_n(\ruleset, \factset)\models \phi(\bar a))$$
 \begin{definition}\label{bdd-def}
A theory \RuleSet has the {\em Bounded Derivation Depth property}\footnote{Or {\em the Finite Unification Set property}, FUS, \cite{BLMS11}.} (in short ``is BDD'')
if:
 \vspace{-2mm}
$$ \forall \phi(y)\in CQ\;\; \exists n_\phi\in {\mathbb N}\;\; \forall \factset \;\;\; Enough(n_\phi,\phi(y),\factset,\ruleset)  $$
%
% for each conjunctive query $\psi(\bar y)$ 
% there exists a number $n_\psi\in \mathbb N$ such that for each instance \FactSet and each 
% tuple\footnote{The length of this tuple must of course be equal to the length of $\bar y$, a condition which is always silently assumed in this context.}  
% $\aaaddd{\factset}$ 
% the following two conditions are equivalent:
% 
% \begin{itemize}
%     \item $Ch(\RuleSet,\FactSet)\models \psi(\bar a)$
%     \item $Ch_{n_\psi}(\RuleSet,\FactSet)\models \psi(\bar a)$
% \end{itemize}
\end{definition}
 
A very important characterization of BDD~\cite{CGL09} is that it is the class of theories which admit query rewriting: 
 
\begin{theorem}\label{thm:bdd-th}
A theory \RuleSet is BDD if and only if for each conjunctive query $\psi(\bar y)$ there exists a finite set $rew(\psi(\bar y))$ of CQs, each of them of the form $\phi(\bar y)$,  such that:

\noindent
\textbullet~ For each instance \FactSet and each tuple $\aaaddd{\FactSet}$ there is $ Ch(\RuleSet,\FactSet)\models \psi(\bar a)$ if and only if there exists 
    $\phi(\bar y ) \in rew(\psi(\bar y ))$ such that  $\FactSet\models \phi(\bar a) $.\\
    \textbullet~ The set $rew(\psi(\bar y))$ is minimal: if $\phi(\bar y)\neq \phi'(\bar y)$ are two  elements of $rew(\psi(\bar y))$ then $\phi(\bar y)$ is not contained in $\phi'(\bar y)$.

 \end{theorem}

\noindent
{\bf The BDD property -- exercises.}
 Now we would like to encourage the  Reader  to solve a few exercises. And
well, we confess that some of the exercises  will be later used as lemmas.
 
\begin{exercise}\label{bdd1}
Consider the theory ${\mathcal T}_p$ consisting of just one rule
$E(x,y) \Rightarrow \exists z\; E(y,z)$.
Notice that this theory is BDD.
\end{exercise}

\noindent{\em Hint:} For given $\factset$ the structure $Ch(\RuleSet,\FactSet)$ comprises an forward-infinite $E$-path starting at each element of the active domain of $\factset$. Notice that if a query with $k$ variables is satisfied in $Ch(\RuleSet,\FactSet)$  then one can satisfy it using only 
terms of $Ch(\RuleSet,\FactSet)$ which are not further than in distance $k$ from the original structure $\factset$.

\noindent{\em Comment:}  This can be easily generalized: it is well known that  {\bf all linear theories are BDD} (a theory is linear if each rule only has one atom in its body).

In Exercises \ref{ex-distancing}--\ref{soon}  $\RuleSet $ is assumed to be  a BDD theory. 
 
 \begin{exercise}\label{ex-distancing}
 If \ruleset is connected then there exists  $d\in \mathbb N$ such that for each \factset and for each two terms $c,c'$ of $dom(\factset)$, if 
 $dist_{Ch(\RuleSet,\FactSet)}(c,c')=1$ then  $dist_\FactSet\leq d$.
 \end{exercise}
 
 \noindent
{\em Comment:} the hidden goal of Exercise \ref{ex-distancing} is to create the intuition of BDD as a ``local'' property: if terms from
$dom(\factset)$ appear in one atom somewhere in $Ch(\RuleSet,\FactSet)$ then they could not possibly be far away from each other already in \factset.
 
\begin{exercise}\label{cw-unique}
For any CQ $\phi(\bar y)$ the set $rew(\phi(\bar y))$ is  unique.
\end{exercise}
\noindent
{\em Hint:} Assume, towards contradiction, that $rew(\phi(\bar y))$ is not unique and use \cref{thm:bdd-th} for a CQ  being in one of such rewritings but not in some other rewriting. 
% 
% 
%  
% \begin{exercise}\label{cw5}
% Let $\psi(\bar y)$ and $\phi(\bar y)$ be two conjunctive queries with the same  free variables. 
% Suppose the following holds:\vspace{-1mm}
% $$\mathop{\forall}\limits_{\FactSet} \;\mathop{\forall}\limits_{\aaaddd{\FactSet}} \FactSet \models \phi(\bar a) \Rightarrow Ch(\RuleSet,\FactSet)\models \psi(\bar a)\vspace{-1mm}$$
% Then there exists a CQ $\phi'(\bar y) \in rew(\psi(\bar y))$ such that $\phi'(\bar y)$ is contained in $\phi(\bar y)$. 
% \end{exercise}
% 
%  \begin{exercise}\label{cw7}
%   Suppose $\phi(\bar y)\in rew(\psi(\bar y))$.
%  Let $\FactSet$ be any instance and  $\aaaddd{\FactSet}$. 
%  Let
%  $Ch(\RuleSet,\FactSet)\models \phi( \bar a)$.  Then there exists  
%  query $\phi'(\bar y)\in rew(\psi(\bar y))$   such that 
%  $\FactSet\models \phi'(\bar a)$.
%  \end{exercise}
%  
%  \noindent
%  {\em Hint:} Notice that $Ch(\RuleSet,Ch(\RuleSet,\FactSet)) = Ch(\RuleSet,\FactSet)$.
%  
 
% to nie jest prawda
% Next exercise shows that if a rewriting $\phi(\bar y)$ of a query $\psi(\bar y)$ is non-connected then if we replace a  component of $\phi(\bar y)$  by its own rewriting  we will get a new rewriting of $\psi(\bar y)$ (or rather a query containing such a rewriting):

\begin{exercise}\label{soon}
There exists a natural number $n_{at}$ (which depends only on \ruleset) such that for any instance $\factset$, for any  $i\in \mathbb N$,
for any $R\in\Predicates$ and for any tuple $\bar t \in dom(Ch_i(\ruleset,\factset))^{\text{\tiny arity(R)}}$ if $Ch(\ruleset,\factset)\models R(\bar t)$ then 
$Ch_{i+n_{at}}(\ruleset,\factset)\models R(\bar t)$.
\end{exercise}

\noindent
Which means that facts about terms are produced by chase soon after the terms are created (with a constant delay).

  \section{The Core Termination Property}\label{sec-core-termination}
 
 \begin{definition}\label{core-t-1}
 Theory \RuleSet {\em ~Core Terminates}\footnote{
  Or has the  {\em Core Termination Property}, or has the {\em Finite Expansion Set Property}, FES.
 } 
 if:
 \vspace{-2mm}
 $$\forall \factset \;\;  \exists n\in {\mathbb N}\;\;  \forall \phi(y)\in CQ \;\;\; Enough(n,\phi(y),\factset,\ruleset)  $$
 % 
%  
%  for each instance 
%  \FactSet there exists
%  $n\in \mathbb N$ such that for each
% CQ $\psi(\bar y)$ and each 
%    tuple $\aaaddd{\FactSet}$
%    if $Ch(\RuleSet,\FactSet)\models \psi(\bar a)$ then 
%    $Ch_n(\RuleSet,\FactSet)\models \psi(\bar a)$.
   \end{definition}

 It is easy to see that one could equivalently define the  Core Termination Property  as:
 
 \begin{definition}\label{core-t-2}
 Theory \RuleSet   Core Terminates   if for each fact set
 \FactSet there is $n\in \mathbb N$ and 
 a homomorphism $h$ from $Ch(\RuleSet,\FactSet)$ to $Ch_n(\RuleSet,\FactSet)$, such that 
 $h$ is the identity on $dom(\factset)$.
   \end{definition}
   
Notice if $h$ is a homomorphism as in \cref{core-t-2}, the image $h(Ch(\RuleSet, \FactSet))$ is a finite structure.
 
% \noindent
% Or as (and again the $n$ will not change):
 
% \begin{definition}\label{core-t-3}
% Theory \RuleSet   Core Terminates  if for each instance 
% \FactSet there exists a number $n\in \mathbb N$  and a set $M$ of facts with
% $\factset \subseteq M\subseteq Ch_n(\RuleSet,\FactSet)$ and  $M\models\ruleset$.
% \end{definition}

 \noindent
 A {\bf \small stronger} notion is the All-Instances Termination Property:
 
 \begin{definition}
 A theory \RuleSet has the  All-Instances Termination Property  if for each instance 
 \FactSet there exists
 a number $n\in \mathbb N$ such that  $Ch(\RuleSet,\FactSet) = Ch_n(\RuleSet,\FactSet)$.
   \end{definition}
   
Intuitively, for a  theory \ruleset to All-Instance Terminate, it is {\bf \small necessary}\footnote{But not sufficient, unless we are talking about restricted chase, which is not the case here.} that for each instance \factset
one of the chase stages  $Ch_n(\ruleset, \factset)$ {\bf \small is}  already a model of \ruleset. And for Core Termination
it is {\bf \small sufficient} that for each instance \factset some stage  $Ch_n(\ruleset, \factset)$ {\bf \small contains} a model of \ruleset.

   \begin{exercise}
   Notice that the BDD theory from Exercise \ref{bdd1} is not Core-Terminating.
   \end{exercise}

\begin{exercise}\label{core1}
Notice that the theory consisting of  two rules:

$E(x,y) \Rightarrow \exists z\; E(y,z)$

$E(x,x'), E(x',x'') \Rightarrow E(x',x')$

\noindent
is Core Terminating but does not All-Instances Terminate.
\end{exercise}

\begin{definition}\label{core-t-345}
Let \ruleset be a Core-Terminating theory and let \factset be an instance. 
 \begin{itemize}
 \item Let $h$ be homomorphism, as in \cref{core-t-2}, of the smallest possible image (if there many homomorphisms whose images are of the same cardinality, then choose any of them). Then  $Core(\ruleset, \factset)$ is defined\footnote{It is known \cite{core-paper} that  whatever $h$ we choose the resulting structure  $Core(\ruleset, \factset)$ will be the same, up to isomorphism. But we are never going to use this fact.} as $h(\Chase{\RuleSet, \FactSet})$. 
 \item By $c_{\ruleset, \factset}$ we will denote the smallest number $n$ such that $Core(\ruleset, \factset) \subseteq Ch_n(\ruleset, \factset)$.

 \end{itemize}
 \end{definition}

  \begin{exercise}\label{core-wlasne}\label{obs:core-square}
 \begin{itemize}
 \item
If $\FactSet\models \ruleset$ then $Core(\ruleset, \factset)=\factset$.

\item
For any instance $\FactSet$ there is $\Core{\ruleset, \Core{\ruleset, \FactSet}} = \Core{\ruleset, \FactSet}$.
\end{itemize}
\end{exercise}

 \section{The FUS/FES Conjecture}\label{fus-fes-conjecture}

  \begin{definition}\label{def:ubdd}
Theory $\RuleSet$ is Uniformly BDD (UBDD) if:
 \vspace{-1mm}
 $$  \exists c_\ruleset\in {\mathbb N}\;\; \forall \factset \;\;  \forall \phi(y)\in CQ \;\;\; Enough(c_\ruleset,\phi(y),\factset,\ruleset)  $$

% there exists  $k_\RuleSet\in \mathbb N$ such that for any CQ $\phi(\bar y)$,
% for any  instance $\FactSet$ and for every tuple $\aaaddd{\factset}$ there is
% 
% $Ch(\ruleset, \factset)\models\phi(\bar a)\Leftrightarrow Ch_{k_\RuleSet}(\ruleset, \factset)\models\phi(\bar a)$.
\end{definition}

It is easy to see that:
 
   \begin{observation}\label{fact:ubdd-by-core}
Theory $\RuleSet$ is UBDD if and only if there exists  $c_\RuleSet\in \mathbb N$ such that for any instance $\FactSet$ there is
$\Core{\RuleSet, \FactSet} \subseteq \Step{c_\RuleSet}{\FactSet}$.
\end{observation}
 
 Notice also that the numbers $c_\RuleSet$ in Observation \ref{fact:ubdd-by-core} and in Definition \ref{def:ubdd} are equal.
 
It is very easy to produce examples of BDD theories which are not UBDD (like in Exercise \ref{bdd1}). But all examples we could produce are not Core-Terminating.
And it is  easy to produce examples of Core-Terminating theories which are not UBDD, but they are not BDD either (all unbounded Datalog theories
will be very happy to serve as examples). This leads to:

\begin{conjecture}[The FUS/FES Conjecture]
Each BDD theory which is Core-Terminating is also UBDD. 
\end{conjecture}
 
 This conjecture was  studied in \cite{LMU16} where an incorrect proof was proposed, and in \cite{BLMTU19}
 where it is proved that it would hold true if the assumption that \ruleset is Core-Terminating was replaced by the assumption that it is All-Instances Terminating.
Notice that the conjecture would be false if infinite theories, over infinite (binary) schemas were allowed:

\begin{example}\label{infinite} Suppose we had a relation symbol $E_i$ for every $i\in \mathbb N$. Let theory \ruleset comprise, for each  $i\in \mathbb N^+$,
the rule $E_i(x,y)\Rightarrow \exists z\; E_{i-1}(y,z)$. Then $\ruleset$ is BDD, Core Terminating, but is not UBDD. To see why this is the case, notice that
 only facts from a finite number of relations can appear in every given finite instance. 
\end{example}

 The main message of this paper is that even if a (finite) counterexample to the Conjecture exists, 
 it is not going to be found 
 in the areas of the BDD class which are anywhere near 
 the known avenues. But we also show that the known avenues only reach a small part of the BDD class.

%%%%%%%%%%%%%%%%%%%%%%%%%%%%%%%%%%%%%%%%%%%%%%%%%%%%%%%%%%%%%%%%%%%%%%%%%%%%%%%%%%%%%%%%%%%%%%%%%%%%%%%%%
% POCZĄTEK ŁATWEGO ROZDZIAŁU O LOKALNOSCI

\section{Local Theories.}\label{local-theories}

For a BDD theory \ruleset and a query $\psi(\bar y)$ by $r\hspace{-0.5mm}s_\ruleset(\psi(\bar y))$ let us denote the maximal number of atoms in a query from 
the set $rew(\psi(\bar y))$ (where $\rsize$ stands for ``rewriting size''). By $\rsize^{at}_\ruleset$ we will denote the maximal
$\rsize_\ruleset(\psi(\bar y))$ for atomic queries $\psi(\bar y)$ (that is queries whose body consists of a single atom).

The following (obvious) observation reflects the intuition that if \ruleset is BDD then \ruleset-mediated query answering is a ``local''
phenomenon:

\begin{observation}
For any BDD theory \ruleset, any conjunctive query  $\psi(\bar y)$, any instance \factset, and any 
tuple $\aaaddd{\factset}$ the two conditions are equivalent:
\begin{itemize}
\item $ Ch(\ruleset, \factset) \models  \psi(\bar a)$
\item There exists an instance $\factseta \subseteq \factset $ with $|\factseta| \leq \rsize_\ruleset(\psi(\bar y))$ and such
that $ Ch(\ruleset, \factseta) \models  \psi(\bar a)$.
\end{itemize}
\end{observation}

Notice that  in particular (if \ruleset is BDD) every atomic query which is true in $ Ch(\ruleset, \factset)$ is also true in 
$\bigcup_{\FactSetAux \subseteq \FactSet, |\FactSetAux| \leq \rsize^{at}_\ruleset} \hspace{0mm} \Chase{\RuleSet, \FactSetAux}$
which is a subset\footnote{Skolem naming convention is important here. Without it it would be unclear what a union of Chases is supposed to mean.} of $ Ch(\ruleset, \factset)$.\\

This leads to:

\begin{definition}\label{def-local-theories}
A theory $\RuleSet$ is \textbf{local} if there exists $l_\RuleSet\in \mathbb N$ such that for every instance $\FactSet$ the following holds:
$$\bigcup_{\FactSetAux \subseteq \FactSet, |\FactSetAux| \leq l_\RuleSet} \; \Chase{\RuleSet, \FactSetAux}\;\; =\;\;\Chase{\RuleSet, \FactSet}$$
\end{definition}

Is a theory local if and only if it is BDD? One implication is easy to show. And, additionally,
we can show that local theories {\bf \small admit linear size rewritings}:

\begin{observation}
Each local theory \ruleset is   BDD. Moreover,
 for each local theory \ruleset and
for each query $\psi(\bar y)$ there is: $\rsize_\ruleset(\psi(\bar y))\leq l_\ruleset|\psi(\bar y)|$ (recall that 
$|\psi(\bar y)|$ is the size of $\psi$).

\end{observation}
\begin{proof}
Take any local theory \ruleset  and any CQ $\psi(\bar y)$. We need to prove that there exists 
 a union of conjunctive queries $\phi(\bar y)$ such that all the disjuncts 
in  $\phi(\bar y)$ are small -- there are at most $l_\ruleset|\psi(\bar y)|$ atoms  in each of them, and 
that for each instance $\factset$ and for  any tuple $\aaaddd{\factset}$
there is  $\Chase{\RuleSet, \FactSet} \models \psi(\bar a)$ if and only if $\factset\models \phi(\bar a)$.

Here is how we construct this $\phi$. Let us, for the rest of the proof, consider only database instances $\factset$  whose all constants are variables,
and such that $\bar y\subseteq dom(\factset)^{|\bar y|}$.
For example $\{R(\bar y, z), R(x ,\bar y), E(x,z)\}$ will be such an instance.

Now, imagine an instance $\factset$ such that  $\Chase{\RuleSet, \FactSet} \models \psi(\bar y)$.
Then it  follows directly from Definition \ref{def-local-theories} that there exists a set $\factseta_\factset \subseteq \factset$ such that 
 $|\factseta_\factset|\leq l_\ruleset|\psi(\bar y)|$ and that  $\Chase{\RuleSet, \FactSeta_\factset} \models \psi(\bar y)$ (this is because 
 each fact of the $|\psi(\bar y)|$ facts in $\psi(\bar y)$ can be produced, by chase, from a set of at most $l_\ruleset$ facts in $\factset$).
 
Now think of $\factseta_\factset$ as of a (quantifier free) conjunctive query: after all it is a set of atoms with variables. And define 
$\phi_\factset(\bar y)$ as $\factseta_\factset$ preceded by existential quantifiers binding all the variables in $\factseta_\factset$ except for 
$\bar y$. While there  are infinitely many possible instances \factset, only  only finitely many  queries $\phi_\factset(\bar y)$ are possible (up to renaming of the bounded variables). Now define our $\phi(\bar y)$ as the disjunction of all possible queries  $\phi_\factset(\bar y)$.
\end{proof}

In Section \ref{proof-for-binary} we show one of the three main  results of this paper which says
that in one important case also the opposite implication is  true:

\begin{theorem}\label{th-binary-local}
If all  predicates that appear in a BDD theory \ruleset are at most binary then \ruleset is local.
\end{theorem}

Our proof can be easily adapted\footnote{To be included in the journal version of this paper.} to also cover all the BDD guarded theories.

%%%%%%%%%%%%%%%%%%%%%%%%%%%%%%%%% POCZĄTEK PORZĄDNEGO ROZDZIAŁU O FUS/FES DLA LOKALNYCH

\section{The FUS/FES Conjecture  for local theories}\label{conjecture-for-local}
The following is another of our three main results:
\begin{theorem}[{\bf A}]\label{thm:local-core-ubdd}
If theory $\RuleSet$ is Core Terminating and local then it is also UBDD.
\end{theorem}

Which means that the FUS/FES conjecture holds for local theories.
Using~\cref{fact:ubdd-by-core} one can reformulate this as:

\addtocounter{theorem}{-1}

\begin{theorem}[{\bf B}]\label{thm:local-core-ubdd-reformulated}
Suppose theory $\RuleSet$  is Core Terminating and local. Then there exists  $c_\RuleSet\in \mathbb{N}$ such that for any instance $\FactSet$ there is
$\Core{\RuleSet, \FactSet} \subseteq \Step{c_\RuleSet}{\RuleSet, \FactSet}$.
\end{theorem}

%\addtocounter{theorem}{-1}
%or as:

%\begin{theorem}[{\bf C}]%\label{thm:local-core-ubdd-reformulated}
%If $\RuleSet$ is Core Terminating and local, then there exists  $c_\RuleSet\in \mathbb{N}$ such that for any instance $\FactSet$ there exists an instance  $M$  satisfying $\factset \subseteq M \subseteq \Step{c_\RuleSet}{\ruleset,\FactSet}$ and  
%$M\models \ruleset$.
%\end{theorem}
%
%\noindent
%Last reformulation follows directly from the equivalence between Definition \ref{core-t-2} and Definition \ref{core-t-3}: if there is any 
%model  $M$ of $\RuleSet$ and $\FactSet$ such that $M \subseteq \Step{c_\RuleSet}{\RuleSet,\FactSet}$
%then also $\Core{\RuleSet,\FactSet}\subseteq \Step{c_\RuleSet}{\ruleset,\FactSet}$.\\

\noindent
{\bf \small 
We devote the rest of this Section to the proof of~\cref{thm:local-core-ubdd-reformulated}.}

\noindent
 {\bf \small Let us fix -- till the end of this Section -- the theory $\RuleSet$,} which is both Core Terminating and local (so also BDD). To simplify the notation, for any instance $\FactSet$ we will write $\Chase{\FactSet}$ instead of $\Chase{\RuleSet, \FactSet}$ and $\Core{\FactSet}$ instead of $\Core{\RuleSet, \FactSet}$.

% 
% \begin{definition}
% Let $n_{\ruleset}\in \mathbb N$ be the greatest of the numbers $n_\alpha\;$ for all atomic queries\footnote{That is quantifier-free queries of the form 
% $P(\bar x)$ where $P\in\Predicates$ and where $\bar x$ is a tuple of variables. There are of course finitely many such queries. The 
% number  $n_\alpha$ is like in Definition \ref{bdd-def}.} $\alpha$.
% \end{definition}

\begin{definition}

For an instance  $\FactSet$ define $\Islands$ as the family of sets $\{\factseta: \factseta\subseteq \factset, |\factseta|\leq l_\RuleSet \}$.
% 
% be a family of sets containing, as its member, each subset of $\FactSet$ that is of size at most $l_\RuleSet$.
Define $\UC = \bigcup_{\FactSetAux \in \Islands} \Core{\FactSetAux}$.
\end{definition}

\begin{lemma}\label{lem:UC-subset-of-chasek}
There exists  $k_\ruleset\in \mathbb N$ depending only on $\RuleSet$ (but not on \factset), such that:
$\UC \subseteq \Step{k_\ruleset}{\FactSet}$.
\end{lemma}
\begin{proof}
The set  $\mathfrak{A} = \set{\FactSetAux \;:\; |\FactSetAux| \leq l_\RuleSet}$ of all instances (over\Predicates) of size at most $l_\RuleSet$ is finite (up to isomorphisms).
Recall that $\RuleSet$ is Core Terminating and let $k_\ruleset=\max\{c_{\ruleset, \factset}: \factset\in \mathfrak{A} \}$.
\end{proof}

 Theorem \ref{thm:local-core-ubdd-reformulated}(B) would be proved if we could  find a homomorphism $\Globalh$ 
from $Ch(\factset)$ to $\UC$ (being the identity on $dom(\factset)$).
And, since $\ruleset$ is Core Terminating, we know that for each $\factseta\in \Islands$ there exists a homomorphism $h_\factseta$ from 
$Ch(\factseta)$ to $\UC$ and we know that $\bigcup_{\FactSetAux \in \Islands} Ch(\factseta)=Ch(\factset)$.
So why cannot we just define  $\Globalh=\bigcup_{\FactSetAux \in \Islands} h_\factseta$? Well, this is because 
the domains of $h_\factseta$ and $h_{\factseta'}$ may overlap (for some $\factseta\neq \factseta'$) and there is no
guarantee that $h_\factseta$ and $h_{\factseta'}$ will agree on the terms which are in both domains\footnote{
If $\Globalh$ could be produced in this way then  $\UC\models\ruleset$ would always hold.
And we found an example (not included here) of  $\factset$ and $\ruleset $ such that $\UC\not\models\ruleset$.}.

But the idea to somehow build a global homomorphism $\Globalh$ using the local homomorphisms $h_\factseta$ 
is not that bad. And the set of facts $\UC$ will indeed prove extremely useful in this context. As we are going to show: 

%%%%%%%%%%%%
\begin{lemma}\label{lem:last-big}
For any instance $\factset$ there exists a homomorphism $\Globalh$ from $Ch(\factset)$ to $Ch(\factset)$,
being the identity on $dom(\factset)$, 
such that for each  $t\in dom(Ch(\factset))$ there is  $\Globalh(t)\in dom(\UC) $.
\end{lemma}

\noindent
{\bf \small Let us first show how Theorem \ref{thm:local-core-ubdd-reformulated} follows from Lemma \ref{lem:last-big}:}

% Recall that we want to find  $c_\ruleset\in \mathbb N$ such that for each \factset there will exist a  fact set $M\models \ruleset$ with  $\factset \subseteq M\subseteq Ch_{c_\ruleset}(\factset)$.

Suppose some \factset is fixed and $\Globalh$ is a homomorphism, as promised by Lemma \ref{lem:last-big}.
We know that for each  $t\in dom(Ch(\factset))$ there is  $\Globalh(t)\in dom(\UC) $ and that 
$\UC \subseteq Ch_{k_\ruleset}(\factset)$. So, one could think, we have 
$\Globalh(\Chase{\factset})\subseteq Ch_{k_\ruleset}(\factset) $.

 But it is not that simple. Lemma \ref{lem:last-big}  tells us that all the terms of 
$ \Globalh(\Chase{\factset})$ will  indeed appear in $Ch_{k_\ruleset}(\factset)$. But it says nothing like that about the atoms 
of $ \Globalh(\Chase{\factset})$: it might be that there are atoms in 
$ \Globalh(\Chase{\factset})$ which, despite having all their terms in $dom(\UC)$ are not themselves in $\UC$. To overcome this little problem recall Exercise \ref{soon} 
and let $c_\ruleset=k_\ruleset+n_{at}$. Then $\Globalh(\Chase{\factset}) \subseteq Ch_{c_\ruleset}(\factset)$.
%%%%%%%%%%%%%%

This means that what remains to be presented in this Section is {\bf \small  the proof  of ~\cref{lem:last-big}:}

 Let us begin from a lemma which we found so surprising that we needed to prove it in two ways to believe it is true. Recall that 
for any instance $\FactSeta$ there exists a homomorphism 
$h_\factseta : \Chase{\FactSeta} \to \Core{\FactSeta}$ being the identity on elements of  $dom(\FactSeta)$.
But we need more than that, and -- as the lemma shows -- we can have more:

\begin{lemma}\label{lem:uber-homomorphism}
For any instance $\FactSeta$ there exists a homomorphism $h_\FactSeta^* : \Chase{\FactSeta} \to \Core{\FactSeta}$ being the identity on 
%elements of 
$dom(\Core{\FactSeta})$.
\end{lemma}

\begin{proof}
Let  $\FactSeta' = \Core{\FactSeta}$. Think of $\factseta'$ as of a new instance (it is of course finite). 
By Definition \ref{core-t-2} we get a homomorphism 
$h: Ch(\FactSeta') \to \Core{\FactSeta'}$ that is identity on  $dom(\FactSeta')$. 

Now notice that $Ch(\FactSeta')=Ch(\factseta)$ (this is from Observation \ref{posrodku}) and
that $\Core{\FactSeta'}$ equals  $\factseta'$ (from Exercise \ref{core-wlasne}). 
 This shows that $h$ is  the required homomorphism $h^*_\FactSeta$.
\end{proof}

\begin{proof} Notice that the restriction of $h_\factseta $ to $dom(\Core{\FactSeta})$ is an onto\\ function (and thus a bijection) from $dom(\Core{\FactSeta})$ to $dom(\Core{\FactSeta})$ (otherwise $h_\factseta(h_\factseta(Ch(\FactSeta)))$ would be a model of $\ruleset$, and a proper subset of $\Core{\FactSeta}$,
contradicting Definition \ref{core-t-345}).  Let now $m=|dom(\Core{\FactSeta})|$ and define  $h_\FactSeta^*=h_\FactSeta^{m!}$ (recall that for any permutatiom $\pi$ of a set of $m$ elements the permutation $\pi^{m!}$ is the identity).
 \end{proof}

The main technical lemma of this Section is Lemma \ref{lem:big-lemma}.

\begin{definition}
Let $\FactSet$ be a set of facts and let $\FactSetAux\subseteq \factset$.  We denote the substructure of $\Chase{\FactSet}$ induced on the set of terms $dom(\Chase{\FactSet}) \setminus (dom(\Chase{\FactSetAux})  \setminus dom(\Core{\FactSetAux}))$ as\footnote{To be hyper-precise we should call this new structure $M_{\factset,\FactSetAux}$, but \factset will be fixed and clear from the context. }  $M_\FactSetAux$.
\end{definition}

In human language this reads as: ``First {\em ban} all the terms 
that appear in $Ch(\factseta)$. Unless they  appear in $Core(\factseta)$, then they are not banned. Then remove, from $Ch(\factset)$, all atoms which dare to mention a 
banned term. And you will get $M_\FactSetAux$.''

\begin{lemma}\label{lem:big-lemma}
For any instance $\FactSet$ and for any  $\FactSetAux\subseteq \factset$ the structure $M_\FactSetAux$ is a model of $\ruleset$ and $\FactSet$.
\end{lemma}

\begin{proof}

Clearly $M_\FactSetAux\models \factset$.
In order to prove that also $M_\FactSetAux\models \ruleset$ take any
$\rho \in \RuleSet$ and any mapping\footnote{Recall Definition \ref{application} here.} $\sigma \in \mappings{\rho, M_\FactSetAux}$.
Notice that of course
$ appl(\rho, \sigma)\in \Chase{\FactSet}$, since $\Chase{\FactSet}$ is by definition closed under rule applications.

The atom $appl(\rho, \sigma)$ will now be mentioned often enough to deserve a shorter name, so we will call it  $\alpha$.

It is now sufficient (and necessary) to prove that 
there exists a homomorphism
 from $ \alpha$ to some atom $\alpha'\in M_\FactSetAux$, which is 
the identity on the $\mathit{fr}(\alpha)$. In other words, we need to show, that if 
a body of the rule $\rho$ matches with  $M_\FactSetAux$ (via matching $\sigma$) then we can find an atom in $M_\FactSetAux$
which can witness that $\rho$ is satisfied in $M_\FactSetAux$.
Such $\alpha'$  needs to have the same terms as $\alpha$ on the frontier positions and may have anything on the positions where 
 the existentially quantified variables were in $head(\rho)$ except that if $\alpha$ had equal terms on two such positions then the respective terms in 
 $\alpha'$ must also be equal.

If $\alpha \in M_\FactSetAux$ then of course $\alpha'=\alpha$. 
So, for the rest of the proof,  assume  that  $\alpha \not\in M_\FactSetAux$. And the only reason for the atom $\alpha$ to be in $\Chase{\FactSetAux}$ but not in $M_\FactSetAux$ is that this atom contains some banned term $t$.

But $\sigma(body(\rho))\subseteq M_\factseta$. So this banned term $t$ must not occur in any of  the
atoms of $\sigma(body(\rho))$, and thus $t\not\in \mathit{fr}(\alpha) $. This means that we can be sure at this point that $\rho$ is  not
 a Datalog rule -- atoms derived with a Datalog rule do not have non-frontier terms.

Term $t$ being a non-frontier term of $\alpha$ means  that 
$\alpha$ is the birth atom (recall Observation \ref{lem:termparentatom}) of $t$ in $\Chase{\FactSet}$. 

But $t\in dom(\Chase{\FactSetAux})$ so  $\alpha$ is also the birth atom of $t$ in  $\Chase{\FactSetAux}$. 
So there exists a mapping  $\sigma_1 \in\mappings{\rho, \Chase{\FactSetAux}}  $ 
(which may, or may not, be equal to $\sigma$)
such that 
$\alpha=appl(\rho, \sigma_1)$ and  such that, for any variable $x\in \mathit{fr}(\rho)$, there is $\sigma(x)=\sigma_1(x)$
(this is Observation \ref{lem:termparentatom} again).

Notice  also that for any variable $x\in \mathit{fr}(\rho)$ there is (*) $\sigma_1(x)\in dom(\Core{\FactSetAux})$. This is because of course 
$\sigma_1(x)\in dom(\Chase{\FactSetAux})$ and we know that $\sigma_1(x)$ it is not banned.

Now let $h^*_\FactSetAux$ be a homomorphism from \cref{lem:uber-homomorphism}. 
Since $\alpha\in \Chase{\FactSetAux}$ we get that 
$h^*_\FactSetAux(\alpha) \in \Core{\FactSetAux}$ and thus $h^*_\FactSetAux(\alpha)\in M_\factseta$. 
Also from \cref{lem:uber-homomorphism} and from (*) we have that $h^*_\FactSetAux(\sigma_1(x))=\sigma_1(x)$ for  any $x\in \mathit{fr}(\rho)$.

This means that we can take $h_\FactSetAux(\alpha)$ as our $\alpha'$.\end{proof}

\begin{lemma}
For any instance $\FactSet$ and any $\FactSetAux\subseteq\factset$ there exists a homomorphism $h^*_{M_\FactSetAux}$ from 
$\Chase{\factset}$ to $\Chase{\factset}$ such that 
for each $t\in dom(\Chase{\factset})$ there is 
$h^*_{M_\FactSetAux}(t)\in dom(M_\FactSetAux)$ 
and for each $t\in dom(M_\FactSetAux)$ there is 
$h^*_{M_\FactSetAux}(t)=t$.

\end{lemma}
\begin{proof}
 Of course (Observation \ref{posrodku}) $Ch(M_\FactSetAux)=Ch(\factset)$. 

From Lemma \ref{lem:big-lemma} we know that $M_\FactSetAux\models \ruleset$. So $M_\FactSetAux=Core(M_\FactSetAux)$. Now use Lemma
\ref{lem:uber-homomorphism} for $M_\FactSetAux$ as the instance.
\end{proof}

Let ${\mathcal H}_\factset$ be the set of all homomorphisms $\Islandh^*$ for $\FactSetAux \in \Islands$.

Each $\Islandh^*\in {\mathcal H}_\factset$  has as its domain the set $dom(\Chase{M_\FactSetAux})$, which is equal to $dom(\Chase{\FactSet})$, and has as its image a subset of this domain.
This means that one can compose such homomorphisms, and the resulting function will also be a homomorphism from 
$\Chase{\FactSet}$ to $\Chase{\FactSet}$  (and it will 
be the identity on $dom(\factset)$, since each $\Islandh$ is). Now the rabbit is going to be pulled  out of the hat:
 let us compose {\bf all}  homomorphisms
$\Islandh\in {\mathcal H}_\factset$,  in any order. Call the resulting (``global'') homomorphism $\Globalh$.

Now recall that the proof of Lemma \ref{lem:last-big} (and thus also of Theorem \ref{thm:local-core-ubdd-reformulated}) will be finished once we can show 
that for each term $t\in dom(Ch(\factset))$ there indeed is  $\Globalh(t)\in dom(\UC) $.

Recall our notion of banned terms. But now $\factseta$ is no longer  fixed:  we  
have (for each $\factseta\in \Islands$) a set $ban_\factseta$ of terms which occur somewhere in $Chase(\factseta)$ but not
in $Core(\factseta)$. 
Now see what each $\Islandh^*\in {\mathcal H}_\factset$ does: it is the identity on all terms, except for the terms in $ban_\factseta$.
And it maps the terms from $ban_\factseta$  into $dom(Core(\factseta))$, which means into $dom(\UC)$.

So take any $t\in dom(Ch(\factset))$ and apply $\Globalh$  to this $t$.
If there is any $\Islandh^*\in  {\mathcal H}_\factset$ such that $\Islandh^*(t)\neq t$ then of course $\Globalh(t)\in dom(\UC) $.
So suppose $\Islandh^*(t)= t$ for each   $\Islandh^*$ and consider any $\factseta_t\in \Islands$ such that 
$t\in dom(Ch(\factseta_t))$. But this implies that $t\in dom(Core(\factseta_t)) \subseteq dom(\UC)$.

%%%%%%%%%%%%%%%%%%%%%%%%%%%%%%%%%%%%%%%%%%%%%%%%%%%%%%%%%%%%%%%%%%%%%%%%%%%%%%%%%%%%%%%%%%%%%%%%%%%%%%%%%%%%%%%%%%%%%%%%%
\section{Slightly beyond locality. Sticky theories.}\label{slightly-beyond}

Unfortunately, our notion of locality fails to characterize the entire BDD class.

\begin{example} Let $E$ be an arity 4 relation and $R$ an arity 2 one. Read $E(a,b,b',c)$ as ``$a$ can see an edge from $b$ to $b'$ colored with colour $c$''
and $R(a,c)$ as ``$a$ thinks $c$ is a color''.
The following one-rule sticky theory \ruleset is not local:
\begin{center}
$E(x,y,y',t),R(x,t')\Rightarrow \exists y''\; E(x,y',y'',t') $
\end{center}
(meaning ``if $x$ can see any edge  from $y$ to $y'$,  
and if she thinks  $t'$ is a color,  then she must also see another edge, from $y'$ to some $y''$, of color $t'$'').

To see that it is indeed non-local suppose that it is, and that $l_\ruleset$ is the constant for $\ruleset$, as in Definition \ref{def-local-theories}.
Now take an instance \factset consisting of $l_\ruleset+1$ atoms: an atom $E(a,b_1,b_2,c_1)$ and of atoms 
$R(a,c_i)$ for $1\leq i \leq l_\ruleset$. It is not hard to see that there are atoms in $Ch(\ruleset, \factset)$ which require all the 
atoms from \factset to be produced.
\end{example}

The only reason however\footnote{Recall that we only consider connected theories.} for sticky theories to be non-local are high-degree vertices, like the $a$ in the example. This 
leads to a natural generalization of the notion of locality:

\begin{definition}\label{bd-local}
A theory $\RuleSet$ will be called \textbf{bounded-degree local} (or bd-local) if for any natural number $k$ there exists a constant $l_\RuleSet(k)$ such that for every instance $\FactSet$ having degree at most $k$, the following holds:
\vspace{-2mm}
$$\bigcup_{\FactSetAux \subseteq \FactSet, |\FactSetAux| \leq l_\RuleSet(k)} \hspace{-3mm} \Chase{\RuleSet, \FactSetAux} = \Chase{\RuleSet, \FactSet}$$
\end{definition}

\noindent
We were unable to show that the FUS/FES conjecture holds  for bounded-degree local theories, but we believe that 
with some additional effort the ideas from Section \ref{conjecture-for-local} could probably  be adapted to work also for such theories. 
And of course they do work if only instances of fixed degree are considered.

It is easy to show that sticky theories are indeed bd-local\footnote{This will be included in the full paper.}.
Which means that {\bf \small all the known decidable BDD classes are bounded-degree local}. Surprisingly,  (unlike local theories) not all 
bounded-degree local theories are BDD\footnote{It would be really nice to have a formalization of the intuitive notion of locality which would 
contain all the known BDD classes and would imply being BDD.}:

\begin{example}
It is very easy to see that the following single-rule theory is bounded-degree local but not BDD:
$E(x,y,z),R(x,z) \Rightarrow R(y,z)$.
\end{example}

But even if not all bd-local theories are BDD, all BDD theories from the known decidable BDD classes are bd-local (and all except for sticky theories are just local). So a natural question 
comes: are there BDD theories which are not local in this generalized sense? We found it quite surprising to realize that the answer is 
positive:

\begin{example}\label{pelzanie}
The following BDD theory $\ruleset_c$ is  not bd-local:\\
$E(x,y)\Rightarrow \exists x',y'\; R(x,y,x',y')$\\
$R(x,y,x',y'),E(y,z)\Rightarrow \exists z'\; R(y,z,y',z')$.\\
To prove  that it is BDD one can notice that if $\ruleset_c,\factset\models \phi(\bar a)$, for some 
$\factset$ 
and some $\aaaddda{\factset}$ 
then $Ch_{|\phi(\bar y)|}\models \phi(\bar a)$. 
In order to prove that it is not 
bd-local consider, for each $n\in \mathbb N$, the instance $\factset_n$ consisting of atoms $E(a_1,a_2)$, $E(a_2,a_3)\ldots$,
$E(a_n,a_1)$. The degree of this instance is 2. And  there are atoms in $Ch_n(\ruleset_c, \factset_n)$ which are not 
in  $Ch_n(\ruleset_c, \factseta)$ for any proper subset $\factseta$ of $\factset_n$.
\end{example}
 
 We were however not able to find an example of a theory which would be {\bf hereditary\footnote{This means that not only the theory  is BDD but also all its subsets are.} BDD} but not bd-local. We think it reasonable to conjecture that there are no such theories.

%%%%%%%%%%%%%%%%%%%%%%%%%%%     KONIEC ROZDZIAŁU O FUS/FES DLA LOKALNYCH

\section{Far beyond locality. Theories without small rewritings.}\label{far-beyond}

As we learned in the previous sections, local theories admit linear size rewritings: for a query $\psi(\bar y)$ the rewriting $\psi_\ruleset(\bar y)$ consists of
queries of size at most $l_\ruleset|\psi(\bar y)|$. It is also easy to see that  backward shy 
theories\footnote{Backward shy theories are defined, in \cite{T13}, as  BDD theories such that, for every query $\psi(\bar y)$ 
if $\phi(\bar y)\in rew(\psi(\bar y))$ then only variables from $\bar y$ can occur more than once in $\phi(\bar y)$.
 Sticky theories are backward shy.} admit linear size rewritings as well and, in consequence, also all sticky theories do. 

This is related to another notion of locality:

\begin{definition}
We call a theory \ruleset  {\em   distancing} if there is $d_\ruleset\in \mathbb N$ such that for any instance \factset,  any
$c,c'\in dom(\factset)$, and any  $n\in \mathbb N$ if
$dist_{Ch(\RuleSet,\FactSet)}(c,c')\leq n$ then  $dist_\FactSet(c,c')\leq d_\ruleset n$.
\end{definition}

Is every BDD theory  distancing? It seems that one can easily 
 ``prove'' by induction,  using Exercise \ref{ex-distancing}, that this is the case.
But such proof would  of course be wrong:
the path from $c$ to $c'$ in $Ch(\RuleSet,\FactSet)$ can possibly lead through atoms not containing any constants from the original instance $\factset$. 
It is equally easy to prove (and, this time, correctly) that:

\begin{observation} If a BDD theory  admits linear size rewritings then it is distancing.
\end{observation}

Which means that all theories from previously known BDD classes are distancing. So do there exist non-distancing BDD theories at all? Do there exist BDD theories 
which do not admit linear size rewriting? The answer is in:

\begin{definition}\label{davids-example-def}
Define $\ruleset_d$ as the theory consisting of the following rules:\smallskip\\
(loop) \hfill $ true \Rightarrow \exists x \; R(x,x), G(x,x)$ \smallskip\\
(pins) \hfill $ \forall x (true  \Rightarrow \exists z,z'\; R(x,z), G(x,z'))$ \smallskip \\
(grid) \hfill $ R(x,x'),G(x,u),G(u,u') \Rightarrow \exists z R(u',z),G(x',z)$

\end{definition}

Our rules are not single-head. An additional ternary predicate\footnote{Recall that BDD theories over a binary language which only have single-head rules 
are local and thus also distancing. } (and 4 additional datalog rules, projecting the new predicate on $G$ and $R$)
could be introduced to equivalently reformulate them as single-head TGDs,
but it is more readable here to have multi-head rules and binary schema. We think of instances over our schema 
(and also of bodies of queries) as graphs, whose 
edges have colors, {\bf red} or {\bf green}.
For given $n\in \mathbb N$ denote by $G^n(x_0, x_n)$ the CQ  $\exists x_1\ldots x_{n-1} G(x_0,x_1), \ldots G(x_{n-1},x_n)$  and let  $R^n(x_0, x_n)$ be defined in an analogous way. Let $\phi_R^n(x,y)$ be a CQ defined as $\exists x',y' R^n(x,x'), R^n(y,y'),G(x',y')$. 

By ${\mathbb G}^n(a, b)$ denote a path, consisting of $n$ green edges, whose first 
vertex is $a$ and last vertex is $b$.

\begin{theorem}\label{davids-example-th} 
(A)
Theory $\rulesetd$ is BDD.\\
(B)
For every $n\in \mathbb N$ there is
$G^{2^n}(x,y)\in rew_{\ruleset_d}(\phi_R^n(x,y))$.

\end{theorem}

\noindent 
{\bf \small The rest of this section and the entire Section \ref{implementacja} are devoted to the  proof of 
Theorem \ref{davids-example-th}}.
Since $\rulesetd$ is fixed, we will use the notation $Ch(\factseta)$ instead of $Ch(\rulesetd, \factseta)$

Assume we know  (A) is true and let us prove Claim (B), which implies that $\rulesetd$ is not distancing. Claim (B)  will follow once we notice that: (i) ${\mathbb G}^{2^n}(a,b)\models \phi_R^n(a,b) $ and
 (ii) if $\factset$ is a proper subset of
     ${\mathbb G}^{2^n}(a,b)$ then $Ch(\factset)\not\models \phi_R^n(a,b) $.

To see why (i) is true it is enough to analyze the example with $n=3$ (\cref{fig:dav-chase}).
To see why (ii) holds  notice that if $\factset$ is a proper subset of  ${\mathbb G}^{2^n}(a,b)$ then $a$ and $b$ are in two different connected components 
of  $\factset$ and, since $\rulesetd$ is connected, they are in two different connected components of  $Ch(\factset)$.

\begin{figure}[h!]
  \includegraphics[width=\linewidth]{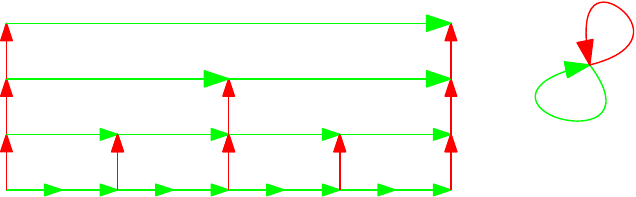}
  \caption{Fragment of  $Ch(\rulesetd, {\mathbb G}^{8}(a_0, a_{8}))$ (please, print in colors!)}
  \label{fig:dav-chase}
\end{figure}

The proof of claim (A) is much more complicated, and, as the following motivating exercise illustrates, the reasons why $\rulesetd$ is BDD are subtle:

\begin{exercise}[Not that easy]
Without rule (loop) theory $\rulesetd$ would not be BDD.
\end{exercise}

Clearly, to prove that $\RuleSetD$ is BDD it is enough to consider only connected queries, since a bound for any non-connected query can be derived from bounds of its connected components. Additionally, notice that due to the rule (loop), if $\phi$ is a boolean query and if 
$\FactSet$ is any instance then $Ch_1(\FactSet) \models \phi$. So for the rest of the proof of Theorem \ref{davids-example-th} (A) 
{\bf \small we will consider only connected non-boolean queries}.
Let us start with: 

\begin{definition}[Marked query]
We define a {\em marked query} as a pair $\langle \phi(\bar y) , V  \rangle$ where $\phi(\bar y)=\exists \bar x \beta(\bar x, \bar y) $ 
is a CQ and $\bar y \subseteq V \subseteq \bar y\cup \bar x$ is a subset of variables of $\phi(\bar y)$.
\end{definition}
\noindent

\noindent
We say that the variables in $V$ are {\em marked}. It will often be more convenient to refer to a marked query without (or before) specifying
its components. The letter $Q$ will be used for that. If $Q= \langle \phi(\bar y) , V  \rangle$ then by 
 $V(Q)$ we  mean $V$, and $q(Q)$ denotes $\phi(\bar y)$.
By $var(Q)$ we  mean the 
set of all variables of $q(Q)$,  and  $Q_{\nic{R}}$ ($Q_{\nic{G}}$) is  the set of red (green) atoms in $q(Q)$.
%Whenever it is not necessary to look inside marked query $\Qm$ we will use $\M$ to denote it.

For a marked query  $Q$ the intention behind the set of marked variables $V(Q)$ is that the variables of $V(Q)$ should not be mapped onto the chase-produced Skolem terms:
\begin{definition}
For a marked query  $Q=\langle \phi(\bar y) , V  \rangle$, a fact set $\FactSet$, and  $\aaaddd{\factset}$  we say that
$Ch(\FactSet)\models Q(\bar a) $ if there exists a homomorphism $h: var(Q) \to dom(Ch(\FactSet))$ witnessing $Ch(\FactSet) \models  \phi(\bar a) $ such that for every 
$v$ there is $h(v)\in dom(\FactSet)$ if and only if $v\in V$.
\end{definition}

Some marked queries are false -- it follows directly from Definition \ref{davids-example-def} that they cannot be 
satisfied in any $Ch(\factset)$:

\begin{observation}\label{proper-in-chase} For any instance \factset 
\begin{itemize}
% \item
% if $R(a,b)$ or $G(a,b)$ is an edge in $Ch(\factset)$ and if $b\in dom(\factset)$  then  $a\in dom(\factset)$; 
% 
% \item if ${\mathcal C}=\{E(a_0,a_1), E(a_1,a_2)\ldots E(a_k,a_0)\}$ is a cycle in  $Ch(\factset)$ (where each $E$ can independently be either $R$ or $G$)
% then ${\mathcal C}\subseteq \factset$;

\item
if $R(a,b)$ or $G(a,b)$ is an edge in $Ch(\factset)$ and if $b\in dom(\factset)$  then  $a\in dom(\factset)$; 

\item if ${\mathcal C}=\{E_0(a_0,a_1), E_1(a_1,a_2)\ldots E_k(a_k,a_0)\}$ is a cycle in  $Ch(\factset)$ (where each $E_i\in\{R, G\}$)
then ${\mathcal C}\subseteq \factset$;

\item if $E(a_1,b)$ and $E(a_2,b)$ are two atoms in  $Ch(\factset)$ (where $E\in\{R, G\}$) and $a_1\in  dom(\factset)$ then 
also $a_2\in  dom(\factset)$.
\end{itemize}
\end{observation}

Then it follows from Observation \ref{proper-in-chase} that:

\begin{observation}\label{proper-in-query} 
Suppose $Ch(\FactSet)\models Q$. Then:
\begin{itemize}
\item[(i)]
if $R(z,z')$ or $G(z,z')$ is an atom in $q(Q)$ and $z'\in V(Q)$ then also $z\in V(Q)$;
\item[(ii)]
if $\{E_0(z_0,z_1), E_1(z_1,z_2)\ldots E_k(z_k,z_0)\}$ is a cycle in   $q(Q)$ (where each $E_i\in\{R, G\}$) then 
$z_i\in V(Q)$ for each $i\in \{0,\ldots k\}$;
\item[(iii)]  if $E(z_1,u)$ and $E(z_2,u)$ are two atoms in  $q(V)$ (where $E\in\{R, G\}$)   and $z_1\in V(Q))$ then 
also $z_2\in  V(Q)$.
\end{itemize}
\end{observation}

% 
% \begin{observation}\label{proper-in-query} 
% Suppose $Ch(\FactSet)\models Q$. Then:
% \begin{itemize}
% \item
% if $R(z,z')$ or $G(z,z')$ is an atom in $q(Q)$ and $z'\in V(Q)$ then also $z\in V(Q)$;
% \item
% if $E(z_0,z_1), E(z_1,z_2)\ldots E(z_k,z_0)\}$ is a cycle in   $q(Q)$ (where each $E$ can independently be either $R$ or $G$) then 
% $z_i\in V(Q)$ for each $i\in \{0,\ldots k\}$;
% \item  if $E(z_1,u)$ and $E(z_2,u)$ are two atoms in  $q(V)$ (where both symbols $E$ are $G$ or both are $R$) and $z_1\in V(Q))$ then 
% also $z_2\in  V(Q)$.
% \end{itemize}
% \end{observation}

Marked queries satisfying conditions (i)-(iii) of Observation \ref{proper-in-query} will be called {\em properly marked}.
Marked queries whose all variables are marked will be called {\em totally marked}. 
Properly marked queries which are not totally marked will be called {\em live}.
Notice that for a totally marked query
 $\langle \phi(\bar a) , V  \rangle$,  an instance $\factset$ and a tuple $\aaaddd{\factset}$ there is
 $Ch(\factset)\models \langle \phi(\bar a) , V  \rangle$ if and only if
$ \factset \models \phi(\bar a)$.

\subsection{High-level proof of claim (A). The process.}

\noindent
Now {\bf \small we get some conjunctive query} $\phi(\bar y)$, which will be fixed till the end of this proof,
and we want to show that there exists a rewriting of $\phi$, i.e. a finite set $rew(\phi(\bar y))$ of CQs
such that for each $\factset$ and  $\aaaddd{\factset}$ there is $Ch(\factset)\models \phi(\bar a)$
if and only if there exists a $\psi(\bar y)\in rew(\phi(\bar y))$ such that $\factset\models \phi(\bar a)$.

Notice that 
 $rew(\phi(\bar y))$  will be constructed if we can produce a finite set $\mathcal S$ of  marked queries, such that:\smallskip\\
($\spadesuit$)
for each $\factset$ and each $\aaaddd{\factset}$ we have $Ch(\factset)\models \phi(\bar a)$ 
if and only if  $Ch(\factset)\models Q(\bar a)$ holds for some   $Q\in \mathcal S$;\smallskip\\
($\clubsuit$)
there are no live queries in $\mathcal S$.\smallskip

${\mathcal S}$ will be constructed as the result of some {\em process}. As the starting point of this process
let ${\mathcal S}_0=\{Q:q(Q)=\phi(\bar y)\}$, the set of all possible markings of $\phi(\bar y)$. One can easily see that  ${\mathcal S}_0$ satisfies Condition ($\spadesuit$)
above, but
there is  no reason to think that it also satisfies Condition ($\clubsuit$). 

Now, the plan is as follows. {\em Five  operations} are going to be defined in Section \ref{implementacja}, called {\em cut-red}, {\em cut-green},
{\em fuse-red}, {\em fuse-green} and {\em reduce}. Each of them will:\\
\textbullet~  take, as an input, some marked query  $Q$;\\ 
\textbullet~  remove from $q(Q)$, one variable and some atoms (operation {\em reduce} will remove one red atom and one green) ;\\ 
\textbullet~  keep the marking of the surviving variables unchanged;\\
\textbullet~  in one case (of operation {\em reduce}) add two new variables, one red atom, and two green atoms);\\
\textbullet~ return the resulting  marked query (except for operation  {\em reduce} which will return four marked queries: one for each 
 of the four possible markings of the two new variables).
 
 It will be shown in Section \ref{implementacja} that:

\begin{lemma}[Completness]\label{completness-lemma}
If a query is live  then at least one of the five operations can be applied to it.
\end{lemma}

At this point we can define our process, which is supposed to ultimately create ${\mathcal S}$. Start from ${\mathcal S}_0$. Once ${\mathcal S}_i$ is 
defined, which does not satisfy condition ($\clubsuit$), take any live query $Q\in {\mathcal S}_i$, 
apply one of the five operations to this query, and define  ${\mathcal S}_{i+1}$ as  ${\mathcal S}_i$ with $Q$ 
replaced by the query (or queries) resulting from this application. Clearly, for the process to make sense we will need to prove, in Section \ref{soundness-lemma-proof} (in the Appendix):

\begin{lemma}[Soundness]\label{soundness-lemma}
Suppose an application of one of the five operations to some marked query $Q_0$ returns a set $\mathcal Q$ consisting of one or several marked queries.
Then, for any $\factset$ and any $\aaaddda{\factset}$ we have
 $Ch(\factset)\models Q_0(\bar a)$ 
if and only if there exists a $Q\in \mathcal Q$ such that $Ch(\factset)\models Q(\bar a)$.
\end{lemma}

It follows from the above lemma that each ${\mathcal S}_i$ satisfies Condition ($\spadesuit$). What we still need to show is that 
the process terminates: as some point we will get  ${\mathcal S}_i$ without live queries. 
For that we are going to use {\em ranks}. 

\subsection{High-level proof of claim (A). Termination.}

 For  a marked query $Q$
 % with $|Q|_{\nic{R}}=n$ 
 and an atom 
 $\alpha\in Q_{\nic{G}}$ the {\bf e}dge  {\bf r}an{\bf k} $erk(\alpha,Q)$ will be 
defined (\cref{def:variable-rank}) as
some natural number, reflecting ``the minimal cost of hiking from a marked variable to $\alpha$''. Then we will prove that:

\begin{lemma}\label{invariant-lemma}
Suppose a marked query  $Q'$ is returned as a result of applying one of the five operations to
 $Q$. Then:
 
 \begin{itemize}
 \item[(i)] \mbox{If the operation is {\bf cut-red} or {\bf fuse-red} then $|Q_{\nic{R}}|>|Q'_{\nic{R}}|$.}
 
 \item[(ii)] If the operation is {\bf cut-green} then $|Q_{\nic{R}}| = |Q'_{\nic{R}}|$ and  for each  
 $ \alpha\in Q_{\nic{G}}$ there is $erk(\alpha, Q)\geq erk(\alpha, Q')$. 
 
 \item[(iii)] If the operation is {\bf fuse-green} then $|Q_{\nic{R}}| > |Q'_{\nic{R}}|$ or $|Q_{\nic{R}}| \geq |Q'_{\nic{R}}|$ 
 and  $erk(\alpha, Q)\geq erk(\alpha, Q')$ for each  $ \alpha\in Q_{\nic{G}}$.
 
 \item[(iv)] If the operation is {\bf reduce} and  if $\alpha\in Q_{\nic{G}}$ is the green atom removed by the operation and 
  if $\alpha'\in q(Q')$ is any of the two green atoms added by the operation then:
  \begin{itemize}
 \item[(a)]  $|Q_{\nic{R}}| = |Q'_{\nic{R}}|$;
\item[(b)]  $erk(\alpha', Q') < erk(\alpha, Q) $;
\item[(c)] 
if  $\beta\in Q_{\nic{G}}\cap Q'_{\nic{G}} $ then $erk(\beta, Q')\leq erk(\beta, Q) $.
\end{itemize}
\end{itemize}
\end{lemma}

%Notice that it follows from 
% Lemma \ref{invariant-lemma} that if ${\mathcal S}_i$ is any set that will at some point appear in our process and if $Q\in {\mathcal S}_i$ then 
% $|Q|_{\nic{R}}\leq K$.

\noindent{\bf Multisets.}  Using Lemma \ref{invariant-lemma} we are going to prove that our process indeed terminates.
To this end we  borrow a technique from the term rewriting community. 
We will use the notation $\{\ldots \}_m$ to denote a  multiset. 
$M(A)$ will denote the family of all finite multisets with elements from $A$.
By $<_m$ we will denote the (strict) mutiset ordering on $M({\mathbb N})$.
By  ${\mathcal R}$ we will mean  the set of all possible pairs $\langle k, A\rangle$ where $k\in \mathbb N$ and $A\in M({\mathbb N})$.
For   $\langle k, A\rangle, \langle k',A'\rangle\in {\mathcal R}$
define $\langle k, A\rangle<_{\mathcal R} \langle k',A'\rangle$ if $k<k'$ or $k=k'$ and $A<_m A'$. Finally, let $<_M$ be the (strict) multiset ordering on  
$M({\mathcal R})$. It is well known  (\cite{DM79}) that ($\heartsuit$) if $A$ is well-ordered then the multiset ordering on $M(A)$ is also a well-ordering. 
So $<_m$ is a well-ordering. In consequence $<_{\mathcal R}$, which is the lexicographic ordering
on the cartesian product of two well-ordered sets is a well-ordering too. And, again using $(\heartsuit)$, we get that $<_M$ is  a well-ordering.

\begin{definition}
\begin{itemize}
\item
 For a marked query  $Q$ define its rank 
$qrk(Q)\in \mathcal R$ as $\langle |Q_{\nic{R}}|, \{erk(\alpha,Q) : \alpha \in Q_{\nic{G}}\}_m \rangle$.

\item
For a set of marked queries $\mathcal S$ define its rank $srk({\mathcal S})\in M({\mathcal R})$ as
the multiset  $\{qrk(Q):Q\in {\mathcal S}   \}_m$.

\end{itemize}
\end{definition}

Now, since the set $M({\mathcal R})$ is well-ordered by $<_M$, to prove termination of our process 
it is enough to show that whenever it produces two subsequent sets ${\mathcal S_i}$ and  ${\mathcal S_{i+1}}$ 
there must be (*)  $srk({\mathcal S_{i+1}})<_M srk({\mathcal S_i})$. But recall that ${\mathcal S_{i+1}}$ is ${\mathcal S_i}$
with one marked query, call it $Q$, replaced  by one of the five operations, with a set $\mathcal Q$ consisting of one or several marked queries. 
So (this is how the multiset ordering works) in order to show (*) it is enough to show that for each $Q'\in \mathcal Q$ we have 
$Q'<_{\mathcal R} Q$. 

But this follows immediately from Lemma \ref{invariant-lemma}, from the definition of the lexicographic 
ordering $<_{\mathcal R}$ (if the 
operation in question is {\em cut-red} or {\em fuse-red}) and from the definition of the multiset ordering  $<_m$ (for the remaining three operations).

\section{Proof of Theorem \ref{davids-example-th} (A). Five operations.}\label{implementacja}\label{far-beyond-dwa}

It easily follows from Observation \ref{proper-in-query} that for every live query 
$Q$
there must exist a {\em maximal} variable $x\in var(Q)$. By this we mean that $x\not\in V$ and
that no atom of the form $E(x,z)$ occurs in $q(Q)$ for $E\in \{G,R\}$. Notice that:

\begin{lemma}\label{obs:dav-three-max}
Let $x$ be a maximal variable of a live query $\langle \phi(\bar y), V\rangle$. Then one of the following condition holds:
\begin{itemize}
    \item[(i)] $x$ occurs in exactly one atom  $E(z,x)$, with $E\in \{G,R\}$;
    \item[(ii)] $x$ occurs in exactly two atoms $R(x_r,x)$ and $G(x_g,x)$ for some $x_r,x_g$;
    \item[(iii)] there exist at least two vertices $z\neq z'$ and $E \in \{G,R\}$ such that $E(z,x)$ and $E(z',x)$ are atoms of $\phi(\bar y)$.
\end{itemize}
\end{lemma}

As promised, now we can define the {\bf \small five operations}.\\ 
Suppose $Q=\langle \phi(\bar y), V\rangle$ is a live query and
$x\in var(Q)$.

% cut-red fuse-red reduce
\begin{definition}[$\cutred$]\label{cut-op-def}
Suppose $x$ is as in Lemma \ref{obs:dav-three-max}(i), with $E=R$.
Define $\cutred(\M,x)$ as $\pair{\phi'(\bar y), V(\M)}$ where $\phi'(\bar y)$ is created from $\phi(\bar y)$ by removing  the sole atom containing $x$.
\end{definition}
\noindent
Operation $\cutgreen$ is defined in an analogous way.

\begin{definition}[$\mergered$]
Let $x$, $z$ and $z'$  be as in Lemma \ref{obs:dav-three-max} (iii), with $E=R$.
Then $\mergered(\M,x,z,z')=$  $\pair{\phi'(\bar y), V}$ where $\phi'(\bar y)$ is $\phi(\bar y)$ with all  occurences of $z'$ renamed\footnote{One could wonder what happens if one of the two variables we unify is in $V$ and the other is not in $V$. But this 
is prohibited by Observation \ref{proper-in-query} (iii). } as $z$. 
\end{definition}
\noindent
Operation $\mergegreen$ is defined in an analogous way.

\begin{definition}[reduce]\label{reduce-def}
Suppose $x$ is as in Lemma \ref{obs:dav-three-max} (ii).
 Let $\phi'(\bar y)$ be a query obtained from  $\phi(\bar y)$
by removing atoms  $R(x_r, x)$ and $G(x_g, x)$ and replacing them with atoms 
 $G(x', x''), G(x'', x_r), R(x', x_g)$ where $x'$ and $x''$ are fresh variables. Then  define $\reduce(\M,x)$ as the set of four marked queries $\pair{\phi'(\bar y), V(\M)}$, $\pair{\phi'(\bar y), V(\M) \cup \set{x'}}$, $\pair{\phi'(\bar y), V(\M) \cup \set{x', x''}}$ and $\pair{\phi'(\bar y), V(\M) \cup \set{x''}}$\footnote{Note that $\pair{\phi'(\bar y ), V(\M) \cup \set{x''}}$ is not properly marked, and thus not live, and it will in no way contribute to our process any more.}.
\end{definition}

Now, Lemma \ref{completness-lemma} easily follows from Lemma \ref{obs:dav-three-max} and from Definitions \ref{cut-op-def}--\ref{reduce-def}.  

The next
 thing left to be proven in this Section is \cref{invariant-lemma}. Notice that claim (i) of the Lemma is now
 obvious: $\cutred$ just removes a single red edge and $\mergered$ merges two red edges into one. 
 In order to prove claims (ii)-(iv) however one needs to work a little bit harder.

\vspace{1.5mm}\noindent
{ \bf Ranks. } We are now going to define the rank $erk(\alpha,Q)$ for a live query $Q$ and an atom  $\alpha\in Q_\nic{G}$. To this end we  consider paths, from some marked variable to $\alpha$,  traversing edges of $q(Q)$ (in both directions). To be more precise: 
% 
% We will treat such paths as lists of atoms. That is, given a path $P$ and its $i$-th edge $E(x,y)$ (where $E \in \set{R,G}$) the following holds:
% \begin{itemize}
%     \item if $E(x,y)$ is traversed ``forward'' then $E(x,y)$ is an $i$-th element of list $P$.
%     \item if $E(x,y)$ is traversed ``backward'' then $E^{-1}(y,x)$ (note that $x$ and $y$ are swapped) is an $i$-th element of list $P$.
% \end{itemize}
% 
% However to define a rank for a variable $v$ of $\M$ we will need a slightly more restrictive definition of a path. This gives a rise to the following definition:

\begin{definition}\label{qpaths}
Given a live query $\M$  an {\em $R$-path}  is a finite sequence $P$   such that:

\noindent
 \textbullet~~each of the elements of $P$ is either $E(t,z)$ or $E^{\nic{-1}}(z,t)$ for some   $E(t,z)$ from $q(Q)$, where $E\in\{G,R\}$ 
(obviously, $E^{\nic{-1}}(z,t)$ means that we traverse $E(t,z)$ backwards);\\
\textbullet~~if $E(s,t)$ and $E'(u,z)$ are two consecutive elements of $P$  then $t=u$ (where $E,E'\in\{G,R,\GmJ, \RmJ \}$);\\
$(\star)$~~if $R(t,z)$ is an atom of $q(Q)$ then only one of $R(t,z)$ and $\RmJ(z,t)$ can appear in $P$ and it can appear at most once. 
   \end{definition}

Notice that  an atom  $G(t,z)$ of $q(Q)$, as well as $\GmJ(z,t)$, can appear any number of times in a $R$-path.  Each $R$-path has its elevation and its cost\footnote{The cost of making a step depends on the current elevation. Current elevation depends on the difference between 
the total ascent  and total descent.}:

\begin{definition}[elevation and cost]\label{elevation} For an empty $R$-path $\emptyset$ we define
 $cost(\emptyset)=0$ and $elev(\emptyset)=3^{|Q_\nic{R}|}$. For a path $P = P'E(x,z)$ we define:

\begin{itemize}
    \item  $elev(P) = elev(P')$ if $E \not\in \set{R, \RmJ}$ 
    \item  $elev(P) = 3\cdot elev(P')$ if  $E = R$
    \item  $elev(P) = \frac{1}{3} \cdot elev(P')$ if $E = \RmJ$
   \item  $cost(P) = cost(P') + elev(P')$ if $E \in \set{G, \GmJ}$ 
    \item  $cost(P) = cost(P')$ if $E \not\in \set{G, \GmJ}$ 
\end{itemize}
\end{definition}

Notice that it follows from condition $(\star)$ of Definition \ref{qpaths} that  $elev(P)$ is always a positive natural number.

\begin{definition}\label{vqpaths} For an atom $\alpha=G(u,u')$ in $ q(Q)$ by an  {\em $\alpha$-hike} we mean an   {\em $R$-path} such that:\\
\textbullet~~if $E(t,z)$ or $E^{\nic{-1}}(t,z)$ is the first atom of $P$ then $t\in V(Q)$;\\
\textbullet~~the last element of $P$ is either $G(u,u')$ or $\GmJ(u',u)$.
\end{definition}

\vspace{-1mm}
\noindent
 For $\alpha\in Q_\nic{G}$ 
we denote the set of all $\alpha$-hikes as
$hikes(\alpha,Q)$.

\begin{definition}\label{def:variable-rank}
For  an atom   $\alpha\in Q_\nic{G}$  its rank is defined as:
$erk(\alpha,Q)=min(\set{cost(P): P \in hikes(\alpha,Q)})$.
\end{definition}

It is easy to see that:

\begin{observation}\label{obs:dav-no-bridge}
For  a marked query $Q$ and for $\alpha=G(u,u')$ in $Q_\nic{G}$, if  $P\in hikes(\alpha,Q)$ and $cost(P) = erk(\alpha, \M)$ then 
any of $\alpha$ or  $\GmJ(u',u)$ can  occur only as the last atom of $P$.
\end{observation}

Now we can finally prove Lemma  \ref{invariant-lemma}.

\vspace{1mm}
\noindent
{\bf Proof of Lemma \ref{invariant-lemma} (ii).} 
Let $Q'=\cutgreen(Q,x)$. 
It follows directly from the construction that $|Q_\nic{R}| \geq |Q'_\nic{R}|$. Now take
any $\alpha\in Q'_\nic{G}$. We need to show that $erk(\alpha, Q) \geq erk(\alpha, Q')$.
Let $z\in var(Q)$ be such that $G(z,x)$ is an atom of $q(Q)$. Consider any  $P\in hikes(\alpha,Q)$. 
 Let $P'$ be a path obtained from $P$ by deleting each occurrence of $G(z,x)$ and $\GmJ(x,z)$. 
 Then $P'\in hikes(\alpha,Q') $ and  $cost(P')\leq cost(P)$. \hfill \qedsymbol

% 
% \begin{lemma}\label{invariant-lemma-mergegreen}
% Let $Q'=\mergegreen(Q,x,z,z')$.
% Then $|Q|_\nic{R}>|Q'|_\nic{R}$ or $|Q|_\nic{R} \geq |Q'|_\nic{R}$ and  $vrk(v, \M) \geq vrk(v, \M')$ for each $v\in var(Q')$.
% \end{lemma}
% \begin{proof}

\vspace{1.5mm}
\noindent
{\bf Proof of Lemma \ref{invariant-lemma} (iii).} 
Let $Q'=\mergegreen(Q,x,z,z')$.

If (*) there exists a variable $u\in var(Q)$ such that atoms $R(z,u)$ and $R(z',u)$ are in $q(Q)$, or that $R(u,z)$ and $R(u,z')$  are in $q(Q)$, then the two red edges merge in $q(Q')$ and  $|Q_\nic{R}| > |Q'_\nic{R}|$. Also, if (**) at least two of the atoms $R(z,z')$, $R(z',z)$, $R(z,z)$, and $R(z',z')$
are in $q(Q)$ then  $|Q_\nic{R}| > |Q'_\nic{R}|$.

So suppose there is neither (*) nor (**). Take any  $\alpha\in Q'_\nic{G}$ and   $P\in hikes(\alpha,Q)$. Consider path $P'$ obtained from $P$ by replacing each occurrence of $z'$ in $P$ with $z$. Obviously, $cost(P) = cost(P')$. 
We now will show that $P'\in hikes(\alpha,Q')$. Clearly, one only needs to worry if condition $(\star)$ of Definition \ref{qpaths} holds. Suppose towards contradiction that it does not. It can only happen when there exist $s,t\in var(Q')$ such that two atoms from $\set{R(s,t), \RmJ(t,s)}$ appear in $P$.

Let us assume that $R(s,t)$ (other cases are analogous) appears twice in $P'$. But of course $R(s,t)$ could not appear twice in $P$, which is an $R$-path. So at least one occurrence of  $R(s,t)$ in $P'$ results from the unification of $z$ and $z'$.

There are two  cases: either (a)  $s=t=z$ or (b) exactly one of $s,t$ equals $z$. So suppose (b) happened and without loss of generality assume that $s = z$. We know that there was $R(z,t)$ and $R(z',t)$ in $P$ and they both unified to  $R(z,t)$ in $P'$. But this would imply (*), leading to a contradiction. The remaining case (a) is that $R(z,z)$ occurs twice in $P'$. But this would need (**) to be true, which is a contradiction again.

So $P'\in hikes(\alpha,Q') $ and  $erk(\alpha,\M) \geq erk(\alpha, \M')$. \hfill \qedsymbol

\vspace{1.5mm}
\noindent
{\bf Proof of Lemma \ref{invariant-lemma} (iv).} Let $Q'$ be any of the marked queries in $\reduce(\M,x)$.
Let  $x_r$ and $x_g$ be variables of $\M$ such that $G(x_g, x)$ and $R(x_r, x)$ are atoms of $q(\M)$. 
Recall that $x', x''\in var(Q')$ are two new variables added by {\em reduce}.

Clearly, $|Q'_\nic{R}| = |Q_\nic{R}|$, since when creating $q(Q')$ from $q(Q)$ we have just replaced one red edge with another.

Let us first show  {\bf claim (c)}. Take some $\beta\in Q_\nic{G} \cap Q'_\nic{G}$ and  $P\in hikes(\beta,Q)$ such that 
$cost(P) = erk(\beta,Q)$. Our goal is to find  $P'\in hikes(\beta,Q')$ such that $cost(P') \leq cost(P)$. 

Obviously, one can assume that $P$ does not contain, as a connected subsequence, $G(x_g, x)\GmJ(x, x_g)$. 
If it did, we could remove such subsequence to get a new  $P$ with lower cost. And, as an R-path, $P$ does not contain
subsequence $R(x_r, x)\RmJ(x, x_r)$ either.

If no atom of $P$ contains $x$ then $P\in hikes(\beta,\M')$ and thus we set $P' = P$. Otherwise it contains exactly one of $R(x_r, x)\GmJ(x, x_g)$ or $G(x_g, x)\RmJ(x, x_r)$ as a connected subsequence.
Suppose  $P = P_0R(x_r, x)\GmJ(x, x_g)P_1$. Let $A = \GmJ(x_r, x'')\GmJ(x'', x')R(x', x_g)$ and consider path $P' = P_0AP_1$. Note that $P'\in hikes(\beta,Q)$. Now we need to show that $cost(P) > cost(P')$. But:\smallskip \\
%\vspace{-2mm}\\
$ cost(P)= cost(P_0)+3\cdot elev(P_0)+3\cdot  elev(P_0)\cdot cost(P_1)$\smallskip\\
%\vspace{-3mm}\\
$ cost(P')= cost(P_0)+2\cdot elev(P_0)+3\cdot elev(P_0)\cdot cost(P_1)$\smallskip\\
%\vspace{1.5mm}
The case when $P = P_0G(x_g, x)\RmJ(x, x_r)P_1$, is similar.\smallskip\\
%\vspace{2mm}
Let us now move to  {\bf claim (b)}. Take $P\in hikes(\alpha,Q)$ such that $cost(P) = erk(\alpha,\M)$. We are going to build a path
$P' \in hikes(\alpha', \M')$ such that $cost(P') < cost(P)$.

There are two cases depending on the last atom of $P$:

First case is when  $P$ is $P_0G(x_g, x)$.
Let then $P'$ be $P_0\RmJ(x_g, x')G(x',x'')$ or $P_0\RmJ(x_g, x')G(x',x'')G(x'',x_r)$ (depending on whether $\alpha'$ is $G(x',x'')$ or $G(x'',x_r)$).
 Then of course $P'\in hikes(\alpha',Q)$
and $cost(P)=cost(P_0)+elev(P_0)$ while
$cost(P') \leq cost(P_0) + \frac{2}{3}\cdot elev(P_0)$.

The second case is that $P = P_0R(x_r,x)\GmJ(x, x_g)$. Here the argument is similar. Now we take $P' = P_0\GmJ(x_r, x'')$ or $P' = P_0\GmJ(x_r, x'')\GmJ(x'', x')$ 
(again depending on whether $\alpha'$ is $G(x',x'')$ or $G(x'',x_r)$). Then  $P'\in hikes(\alpha',Q)$
and $cost(P)=cost(P_0)+3\cdot elev(P_0)$ while
$cost(P') \leq cost(P_0)+2\cdot elev(P_0)$ . 
% 
% 
% Here the argument is similar. Now we take $P_{x'} = P_0\RmJ(x_g, x')$ and $P_{x''} = P_0\RmJ(x_g, x')G(x', x'')$. The only ``red'' atom that might violate condition $(\star)$ of Definition   \ref{qpaths} is $\RmJ(x_g, x')$. But in both $P_{x'}$ and $P_{x''}$ it occurs only once. Again from~\cref{obs:dav-no-bridge} we know that $x$ does not occur in any atom of $P_0$, thus $P_{x'} \in hikes(x', \M')$ and $P_{x''}\in hikes(x'', \M')$. Finally, observe that:\smallskip\\
% \vspace{2.5mm} $cost(P_{x'}) = cost(P) - \frac{2\cdot elev(P_0)}{3}$ and\\
% \vspace{1.5mm} $cost(P_{x''}) = cost(P) - \frac{elev(P_0)}{3}$. 
\hfill \qedsymbol

%%%%%%%%%%%%%%%%%%%%%%%%%%%%%%%%%%%%%%%%%%%%%%%%%%%%%%%%%%%%%%%%%%%%%%%%%%%%%%%%%%%%%%%%%%%%%%%%%%%%%%%%%%%%%%%%%%%%%%%%%%%%%%%%%%%%%%%%
%%%%%%%%%%%%%%%%%%%%%%%%%%%%%%%%%%%%%%%%%%%%%%%%%%%%%%%%%%%%%%%%%%%%%%%%%%%%%%%%%%%%%%%%%%%%%%%%%%%%%%%%%%%%%%%%%%%%%%%%%%%%%%%%%%%%%%%%

\section{Theories which are extremely non-local} \label{ojojoj}

\noindent
Theorem \ref{davids-example-th} begs for generalization. For $K\in \mathbb N$ define $\ruleset_d^K$ as
a theory, over  $\Sigma_K=\{I_K, I_{K-1}, \ldots I_1\}$ (each $I_k$ is a binary relation symbol),
comprising, for each $1\leq i< K$ and each $1\leq k\leq K$ the following $2K+1$ rules:

\noindent
($loop$) \hfill $true \Rightarrow \exists x \;\; I_K(x,x), I_{K-1}(x,x), \ldots I_1(x,x)$

\noindent
($pin_k$) \hfill $ \forall x\; (\; true \Rightarrow \exists z \; I_k(x,z) \;)$

\noindent
\mbox{($grid_i$)~$ I_{i+1}(x,x'), I_i(x,u),I_i(u,u')  \Rightarrow \exists z\; I_{i+1}(u',z),I_i(x',z)$}

Using the ideas from Sections \ref{far-beyond}-\ref{far-beyond-dwa} one can show\footnote{We only have room here to outline the proof very briefly. Details will be given in the journal version of this paper.} that:

\begin{theorem}\label{extreme}
For each $K\in \mathbb N$:\\
\noindent {\bf A.}  theory $\ruleset_d^K$ is BDD;\\
\noindent {\bf B.} there is a query $\psi(y,y')$  such that $rew_{\ruleset_d^K}(\psi(y,y'))$ contains a CQ of size $(K$-$1)$-fold exponential in
the size of $\psi$.
\end{theorem}

Claim {\bf B} is relatively easy to prove. For the proof of claim {\bf A} properly marked queries first need to be slightly redefined (but let us skip it here). Then  the 
{\em five operations} need to be generalized in the natural way: we will now have $K$ {\em cut} operations, $K$ {\em fuse} operations, and $K$-$1$ {\em reduce} operations. 

The non-obvious part is how to modify the ranks $erk$ and $qrk$ so that they do their job correctly in the new circumstances. 
For that $I_i$-paths need to be defined (for $1\leq i< K$), analogous to $R$-paths in Definition \ref{qpaths}. But now the condition ($\star$) will apply to atoms of the relation $I_i$
(including $I_i^{\nic{-1}}$). Then $i$-elevation ($elev_i$) will be defined, like in Definition \ref{elevation}. And finally, we will need $cost_i$ of a path, calculated almost like the 
cost of the path in Definition \ref{elevation}: for a path $P = P'E(x,z)$ we have $cost_i(P)=cost(P')+elev_i(P')$ if $E=I_{i-1}$ or $I_{i-1}^{\nic{-1}}$ and
$cost_i(P)=cost(P')$ otherwise. Notice that $I_{i}$ is the new red and $I_{i-1}$ is the new green. Note also that we have a new situation now: $E$ may very well be neither ''green`` now nor ``red''. And that is fine, in such case it neither contributes to the elevation nor to the cost of the path.

Having the function $cost_i$, rank $erk(\alpha)$ of an atom $\alpha$ of the relation $I_{i-1}$ is (like in Section \ref{far-beyond-dwa}) defined as minimal $cost_i$ of
an  $I_i$-path from some marked variable to $\alpha$, and rank $qrk_i(Q)$ is the multiset of all ranks $erk(\alpha)$ of all atoms $\alpha$ of $I_{i-1}$ in $Q$.
Finally, $qrk(Q)$ is the tuple:
\vspace{-1.5mm}
$$\langle |Q_K|,qrk_K(Q), |Q_{K-1}|,qrk_{K-1}(Q),\ldots    |Q_2|,qrk_2(Q)    \rangle \vspace{-1mm}$$
\noindent
where $|Q_i|$ is the number of the atoms of the relation $I_i$ in $Q$. Clearly, the lexicographic ordering on the set of such ranks is a well ordering. 
Now a careful case inspection shows that each of the $3K$-$1$ operations decreases the rank of a query.

%%%%%%%%%%%%%%%%%%%%%%%%%%%%%%%%%%%%%%%%%%%%%%%%%%%%%%%%%%%%%%%%%%%%%%%%%%%%%%%%%%%%%%%%%%%%%%%%%%%%%%%%%%%%%%%%%%%%%%%%%%%%%%%%%%%%%%%%
%%%%%%%%%%%%%%%%%%%%%%%%%%%%%%%%%%%%%%%%%%%%%%%%%%%%%%%%%%%%%%%%%%%%%%%%%%%%%%%%%%%%%%%%%%%%%%%%%%%%%%%%%%%%%%%%%%%%%%%%%%%%%%%%%%%%%%%%

\bibliography{references-carral} 
\bibliographystyle{ieeetr}
%%%%%%%%%%%%%%%%%%%%%%%%%%%%%%%%%%%%%%%%%%%%%%%%%%%%%%%%%%%%%%%%%%%%%%%%%%%%%%%%%%%%%%%%%%%%%%%%%%%%%%%%%%%%%%%%%%%%%%%%%%%%%%%%%%%%%%%%
%%%%%%%%%%%%%%%%%%%%%%%%%%%%%%%%%%%%%%%%%%%%%%%%%%%%%%%%%%%%%%%%%%%%%%%%%%%%%%%%%%%%%%%%%%%%%%%%%%%%%%%%%%%%%%%%%%%%%%%%%%%%%%%%%%%%%%%%

\newpage

\section{Appendix A: Proof of Theorem \ref{th-binary-local} }\label{proof-for-binary} %Theorem \ref{th-binary-local} 

% 
% \section{Appendix A: Proof of Theorem \ref{th-binary-local} }\label{proof-for-binary}

Let us now fix a binary signature $\Signature$ and a BDD theory $\RuleSet$ over $\Signature$. By $\DLRuleSet$ we will denote the datalog rules of $\RuleSet$ and by $\EXRuleSet$ its existential rules.

First we will need to distinguish, among all elements of $\EXRuleSet$, {\em detached}\footnote{Such rules are called {\em disconnected} in \cite{disconnected-rules}, however we think that calling those rules {\em detached} might help the reader to distinguish those from rules that have disconnected bodies.} rules that are of the form $\phi(\bar{x}) \rightarrow \exists {y,z} \psi(y,z)$ or $\phi(\bar{x}) \rightarrow \exists {y} \psi(y)$, that is, rules having empty frontier. Note that when firing a detached rule, the newly created atom, has no common terms with the rest of the Chase. Notice that, since we only consider binary schemas, the non-empty frontier of an existential rule always consists of exactly one variable\footnote{Actually, the assumption we really use in the proof of Theorem \ref{th-binary-local} is not that relations are at most binary, but that the existential rules are ``frontier one''. }. Rules from $\EXRuleSet$ which are not detached will be called {\em sensible}

Clearly, whatever fact set $\FactSet$ we consider, the structure $\Chase{\RuleSet, \FactSet}$ is a disjoint union of three sets of atoms. One set consists of the original facts from $\FactSet$. Second contains {\em existential atoms}, that is facts created in the process of the chase, by rules of $\RuleSet_\exists$. The third set consists of atoms that are created by the rules of $\RuleSet_{DL}$ and will be called as {\em datalog atoms}. We will denote the set of existential atoms of $\Chase{\RuleSet, \FactSet}$ together with atoms of $\FactSet$ with $\ChaseEX{\RuleSet, \FactSet}$.

Let us now concentrate on the structure of $\ChaseEX{\RuleSet, \FactSet}$. Notice that there are again two kinds of atoms there: {\em detached} atoms, created by detached rules and {\em sensible} atoms, created by sensible rules. Notice also that our taxonomy of atoms implies a taxonomy of the terms of $\ChaseEX{\RuleSet, \FactSet}$  
(that is the elements of $dom(\Chase{\RuleSet, \FactSet})\setminus dom(\factset)$): there are sensible terms, created by sensible rules and detached terms, created by detached rules. The set of detached terms will be called $det(\Chase{\RuleSet, \FactSet})$.

\begin{observation}\label{forest} The graph whose vertices are the terms from $dom(\Chase{\RuleSet, \FactSet})$ and whose edges are sensible atoms of $\ChaseEX{\RuleSet, \FactSet}$ is a forest. The set of the roots of the trees of this forest is equal to
$dom(\factset)\cup det(\Chase{\RuleSet, \FactSet})$. The out-degree of the vertices of this forest is bounded by the number of existential rules in $\ruleset$.
\end{observation}

For any $a\in dom(\factset)\cup det(\Chase{\RuleSet, \FactSet})  $  let $\Skeleton(a)$ be the set of all atoms of $\ChaseEX{\RuleSet, \FactSet}$ which are edges of the tree rooted in $a$. Following the naming convention, we name trees that are rooted in detached terms as {\em detached trees}. 

%Since we consider Skolem chase, the following observation holds:
%\begin{observation}
%There is a constant $w\in \mathbb N$, which depends on $\RuleSet$ but not on $\FactSet$, such that the number of detached atoms in 
%$\Chase{\RuleSet, \FactSet}$ is at most $w$. The same is true for the number of detached terms, albeit with different constant.
%\end{observation}

\subsection{First (failed) attempt at the Crucial Lemma}

For any given set of facts $\FactSet$ let {\em parent function} $\ParentFunction_{\RuleSet}$ be any function from $\Chase{\RuleSet, \FactSet} \setminus \FactSet$ to the power set of $\Chase{\RuleSet, \FactSet}$ such that for any atom $\alpha\in \Chase{\RuleSet, \FactSet}$ there exists a rule $\rho$ and a mapping $\sigma$ satisfying:
\begin{itemize}
    \item $\alpha = appl(\rho, \sigma)$,
    \item $\sigma(body(\rho)) = \ParentFunction_{\RuleSet}(\alpha)$.
\end{itemize}

Parent function says which tuple of atom led to the creation of $\alpha$.
Note that there may be more than one such function as $\alpha$ could be created in more than one way during the chase.

Let $\FactSet$ be any set of facts and let $\ParentFunction_{\RuleSet}$ be some parent function. Then we define an {\em ancestor function} $\AncestorFunction_{\RuleSet}$ as follows: 
\begin{itemize}
    \item $\AncestorFunction_{\RuleSet}(\alpha) = \set{\alpha}$ for an atom $\alpha \in \FactSet$,
    \item $\AncestorFunction_{\RuleSet}(\alpha) = \bigcup_{\alpha' \in \ParentFunction_{\RuleSet}(\alpha)} \AncestorFunction_{\RuleSet}(\alpha')$ for other atoms of $\Chase{\RuleSet, \FactSet}$.
\end{itemize}

Intuitively the set $\AncestorFunction_{\RuleSet}(\alpha)$ consists of facts from $\FactSet$ 
which were used during the chase, to prove $\alpha$. Of course there might be more than one ancestor function for any given set of facts as that functions is strictly associated with a particular parent function. This freedom in taking parents, and so in picking ancestors, leads to some problems as we will soon discover.

We would be one step from proving Theorem \ref{th-binary-local} if we had:

\begin{lemma}[Crucial lemma, first attempt, false]\label{false-crucial-lemma}
There is a natural number $M$, depending on rule set $\RuleSet$ 
but not on fact set $\FactSet$, such that for every ancestor function $\AncestorFunction_{\RuleSet}$, for each constant and any detached term $t$ in $\Chase{\RuleSet, \FactSet}$
it holds that:

$$\left\vert \bigcup\nolimits_{\alpha \in \Skeleton(t)} \AncestorFunction_{\RuleSet}(\alpha) \right\vert \leq  M$$

\end{lemma}

 But that Lemma  is unfortuately not true. For a counterexample see: 

\begin{example}\label{roznirodzice} 
Let $\RuleSet$ consist of two rules:
\begin{itemize}
    \item \tgdl{$E(x,y), R(z,y)$}{$\exists v\; E(y,v)$}
    \item \tgdl{$E(x,y), P(z)$}{$R(z,y)$}
\end{itemize}

Suppose  that $M$ as in the Lemma exists and that $\FactSet$ consists of  atom $E(a_0,a_1)$ and atoms $P(b_i)$ for each $1\leq i \leq M$.

Then  $\Chase{\RuleSet, \FactSet}$ will consist of a infinite number of new  facts $E(a_1,a_2)$, $E(a_2,a_3)$,$E(a_3,a_4)\ldots$. In order to build them, however, some facts about relation $R$ will need to be proven by the second rule using a number of $P$ atoms from
$\FactSet$. And it might happen (Skolem/Semi-Oblivious chase is non-deterministic in that aspect) that the proven $R$-facts will be $R(b_1,a_1)$, $R(b_2,a_2)$, \ldots $R(b_M,a_M)$, meaning that 
$\Skeleton(a_1)$ uses during its creation all the $M+1$ facts of $\FactSet$. But of course this is for irrelevant reasons: the same Chase could be built if $P(b_1)$ was used each time some $P$ was needed. 
\end{example}

\subsection{The normalization of $\RuleSet$}
 In order to circumvent the  problems highlighted by Example \ref{roznirodzice} we will
now transform the rule set $\RuleSet$
into
another rule set $\NormRuleSet$ which, apart from some other useful properties, will satisfy the equality
$(*)\;\;\; \ChaseEX{\NormRuleSet, \FactSet} = \ChaseEX{\RuleSet, \FactSet}$.

First we will  define the signature of $\NormRuleSet$. Let us take a fresh set of nullary predicates $\mathbf{M} = \set{M_\phi \;|\; \phi \text{ is a boolean CQ over } \Signature}$. Then our new signature\footnote{We briefly forget here about our promise that signatures would be finite.} $\Signature'$ is $\Signature \cup \mathbf{M}$.

Two procedures will be used during the normalization: {\em body rewriting} and {\em body separation}.

\begin{definition}[Body rewriting]
Let $\rho$ be some rule with body $\beta({\bar x}, {\bar y})$ over $\Signature$ and a head $\gamma({\bar y})$ 
which is an atom from
 $\Signature'$ possibly preceded with the existential quantifier (or two). Then by $Rew(\rho)$ we denote the  set:
$$\set{\beta'({\bar x}, {\bar y}) \Rightarrow \gamma({\bar y}) \;:\; \beta'({\bar x}, {\bar y}) \in rew_{\RuleSet}(\exists {\bar x} \; \beta({\bar x}, {\bar y}))}$$
\end{definition}
% 
% It is important to note, that in order to get the rewriting of the body of $\rho$ it needs to use the same signature $\Signature$ that $\RuleSet$ use. However we are not necessarily interested at the moment how the head of $\rho$ looks like. That is, we do not care whether $\rho$ is Datalog rule or existential rule and if the head contains a nullary predicate or not. In the future we will use this procedure in both the first step of normalization and the third (last) step and that is why we use a somewhat general definition.

While body rewriting can be applied both to existential rules and datalog rules,
the second procedure will only be applied to existential rules. It separates the disconnected
fragment of body of given existential rule and ``encapsulates'' that fragment in
a single nullary predicate from $\mathbf{M}$. That is, given a rule, the procedure  returns a pair of rules. 
One being almost the original rule, but with its body changed to consist of nullary 
predicate and a connected conjunction of atoms. 
The second being a rule whose job is to prove the aforementioned nullary predicate.

\begin{definition}[Body separation]
If $\rho$ is an existential rule of the form $\beta({\bar x}, {\bar y}) \wedge \phi({\bar z}) \Rightarrow \exists \;{\bar u} \gamma({\bar y}, {\bar u})$ such that:
\begin{itemize}
    \item $({\bar x} \cup {\bar y}) \cap {\bar z} = \emptyset$,
    \item $\beta({\bar x}, {\bar y})$ is connected,
   % \item ${\bar y}$ is the frontier of $\rho$,
\end{itemize}
then
%we denote with $sep(\rho)$ a set containing the following two rules:
\begin{itemize}
    \item $ sep_{cc}(\rho) = \beta({\bar x}, {\bar y}) \wedge M_\phi \Rightarrow \exists {\bar u}\; \gamma({\bar y}, {\bar u})$
    \item $ sep_{M}(\rho) =  \phi({\bar z}) \Rightarrow M_\phi$
\end{itemize}

If the body of $\rho$ is connected, we assume that  $\phi({\bar z})$ is empty. And we have a nullary predicate $M_\emptyset\in \Sigma'$ for this occasion.
\end{definition}

The normalization algorithm is performed in  three steps:

%\begin{definition}[Normalization Algorithm]~~~~~~~~~
\medskip
\noindent
\begin{center}
{\sc Normalization Algorithm}
\rule{\linewidth}{0.4pt}
\end{center}

{\sc \underline{Step One}:} $\RuleSet_{I}=\bigcup_{\rho\in \EXRuleSet } Rew(\rho)$\\
% 
% \begin{itemize}
%     \item Let $\RuleSet_{I} \from \emptyset$
%     \item For every existential rule $\rho \in \RuleSet$ do:
%     $$\RuleSet_{I} \from \RuleSet_{I} \cup rew_\ruleset(\rho)$$
% \end{itemize}

{\sc \underline{Step two}:} $\RuleSet_{II}=\{sep_{cc}(\rho) : \rho\in \RuleSet_{I}\}$\\

{\sc \underline{Step three}:}
$\RuleSet_{III}=\bigcup_{\rho\in \RuleSet_{I} } Rew(sep_M(\rho))$\\

{\sc \underline{Return}:} $\NormRuleSet=\RuleSet_{II}\cup \RuleSet_{III}$

% \begin{itemize}
%     \item Let $\RuleSet_{III} \from \emptyset$
%     \item For every rule $\phi \Rightarrow M_\phi$ contained in  $\RuleSet_{II}$ do: 
%     $$\RuleSet_{III} \from \RuleSet_{II} \cup rew_\ruleset(\rho)$$
%     \item Let $\RuleSet_{II, \exists{}}$ be a set of existential rules of $\RuleSet_{II}$
%     \item Let $\RuleSet_{I,c\exists{}}$ be a set of connected existential rules of $\RuleSet_{I}$
%     \item Return $\RuleSet_{II, \exists{}} \cup \RuleSet_{III} \cup \RuleSet_{I,c\exists{}}$  as $\NormRuleSet$ and return $\RuleSet_{I}$, $\RuleSet_{II}$ and $\RuleSet_{III}$ for further analysis.
%     
% \end{itemize}

\vspace{-3mm}
\rule{\linewidth}{0.4pt}
\smallskip

The normalization allows us to attack the source of the problem highlighted in the previous section by separating the ``disconnected ancestors" required by existential rules and encapsulating those ancestors within rules producing nullary predicates.

\begin{observation} \label{obs:detachedrulenullbody}
Let $\Rule$ be a rule that creates a detached atom in $\Chase{\NormRuleSet, \FactSet}$ 
for some fact set $\FactSet$. Then  $\Rule$ is a rule from $\ruleset_{II}$ and its body consist of a single nullary atom.
\end{observation}

\begin{lemma}\label{lem:skeleton_equivalence}
For any set of facts $\FactSet$ over $\Sigma$:
$$ \ChaseEX{\RuleSet, \FactSet} = \ChaseEX{\NormRuleSet, \FactSet} $$
\end{lemma}

\subsection{Proof of Lemma \ref{lem:skeleton_equivalence}}

This entire subsection is devoted to the proof of Lemma \ref{lem:skeleton_equivalence}, which  will follow directly from 
Lemma \ref{lem:AB} and Lemma \ref{lem:skeleton-D}. But first, as a warm-up notice that:

%%%%%%%%%%%%%%%%%%%%%%%%%%%%%%%%%%%%%%%%%% CWICZENIE 78

 \begin{exercise}\label{cw3}
Let $\FactSet$ be any instance and let  $\aaaddd{Ch(\RuleSet,\FactSet))}$. 
 Let  $\phi(\bar y)\in rew(\psi(\bar y))$ and suppose
 $Ch(\RuleSet,\FactSet)\models \phi( \bar a)$. Then  $Ch(\RuleSet,\FactSet)\models \psi( \bar a)$.
 \end{exercise}

  \noindent
 {\em Hint:} Recall that $Ch(\RuleSet,Ch(\RuleSet,\FactSet)) = Ch(\RuleSet,\FactSet)$.

 \begin{exercise}\label{cw7}
  Suppose $\phi(\bar y)\in rew(\psi(\bar y))$.
 Let $\FactSet$ be any instance and  $\aaaddd{\FactSet}$. 
 Let
 $Ch(\RuleSet,\FactSet)\models \phi( \bar a)$.  Then there exists  
 query $\phi'(\bar y)\in rew(\psi(\bar y))$   such that 
 $\FactSet\models \phi'(\bar a)$.
 \end{exercise}
 
 \noindent
 {\em Hint:} Notice that $Ch(\RuleSet,Ch(\RuleSet,\FactSet)) = Ch(\RuleSet,\FactSet)$.

\begin{exercise}\label{cw10}
Suppose $(\exists \bar t, \bar u \;\phi(\bar t, \bar y)\wedge  \gamma(\bar u))\in rew_\ruleset(\exists x\beta(\bar x, \bar y))$ for some CQs
$\phi$, $\gamma$ and $\beta$, where the tuples of variables $\bar t\cup \bar y$ and $\bar u$ are disjoint. 
And suppose $\exists \bar v\;\zeta(v) \in rew_\ruleset(\exists u\; \gamma(\bar u)) $. Then there exists a query 
$\exists \bar z\; \varrho(\bar z, \bar y)\in rew_\ruleset(\exists x \beta(\bar x, \bar y))$ which is contained in 
$(\exists \bar t, \bar v \;\phi(\bar t, \bar y)\wedge  \zeta(\bar v))$.
% 
% 
% Let us take any conjunctive query of the form $\phi(\bar y ) \wedge \Gamma$ from the rewriting $rew_\ruleset(\psi(\bar y))$ of some CQ $\psi(\bar y )$.
% Suppose that queries $\phi(\bar y )$ and $\Gamma$ do not share variables and that $\Gamma$ is a boolean query. If $\Upsilon \in rew_\ruleset(\Gamma)$ then there exist queries $\Gamma'$ and $\phi'(\bar y)$ such that:
% \begin{itemize}
%     \item $\Gamma'$ is contained in $\Upsilon$
%     \item $\phi'(\bar y )$ is contained in $\phi(\bar y)$
%     \item $\Gamma' \wedge \phi'(\bar y)\in rew_\ruleset(\psi(\bar y))$.
% \end{itemize}
\end{exercise}

% \begin{proof}
% We will show that If $\FactSet \models \Upsilon \wedge \phi(\bar y )$ then $\Chase{\RuleSet, \FactSet} \models \Gamma \wedge \phi(\bar y )$ and using \cref{cw5} we will conclude what is required. Let us take any instance $\FactSet$ such that $\FactSet \models \Upsilon \wedge \phi(\bar y )$.
%     
% Now we use the assumption that $\Upsilon \in rew_\ruleset(\Gamma)$ to see that indeed $Ch(\RuleSet,\FactSet)\models \Gamma \wedge \phi(\bar y )$. 
% \end{proof}

%%%%%%%%%%%%%%%%%%%%%%%%%%%%%%%%%%%%%%%%%%

Notice  that the only existential  rules in $\NormRuleSet$ are the ones in $\RuleSet_{II}$.
Notice also that the only datalog rules in $\NormRuleSet$ are the ones in $\RuleSet_{III}$ and hence the only 
atoms which are in $\Chase{\NormRuleSet, \FactSet} $ but not in $ \ChaseEX{\NormRuleSet, \FactSet} $ are the nullary
atoms from the set $\mathbf{M}$.

By $\ChaseEXi{\RuleSet, \FactSet}{i}$ we denote the intersection of $\Step{i}{\RuleSet,\FactSet}$ and $\ChaseEX{\RuleSet, \FactSet}$.

\begin{lemma}\label{lem:AB}
For each fact set $\FactSet$ over $\Sigma$ and each  $k\in \mathbb N$:
$$ \ChaseEXi{\RuleSet_{NF}, \FactSet}{k}\subseteq \ChaseEX{\RuleSet, \FactSet}.$$
\end{lemma}
 
\begin{proof}
We will show this by induction on $k$. The case for $k=0$ is trivial. Assume the claim is true for some $k\in \mathbb N$. We are going to show it is also true for $k+1$. 

Let $\alpha $
%= appl(\Rule_\alpha, \sigma_\alpha)$ 
be an atom from $\ChaseEXi{\NormRuleSet, \FactSet}{k+1}$, which is not in $\ChaseEXi{\NormRuleSet, \FactSet}{k}$.
Then $\alpha = appl(\Rule_\alpha, \sigma_\alpha)$ for some:\\

\noindent
{\bf --} $\Rule_\alpha\in \RuleSet_{II}$, the rule which actually created $\alpha$, of the form 
     \hspace*{3mm}$\phi(\bar t, \bar y)\wedge M_\gamma \Rightarrow \exists \bar z\; \alpha_0(\bar y,\bar z)$,\\
     {\bf --} $\sigma_\alpha$ such that  $\sigma_\alpha(\phi(\bar t, \bar y) )\subseteq \ChaseEXi{\NormRuleSet, \FactSet}{k}$, \hfill $(\heartsuit)$\\
 {\bf --}   $M_\gamma$ such that $\ChaseEXi{\NormRuleSet, \FactSet}{k}\models M_\gamma$.\\

%Let $\rho_\alpha$ be the rule of $\NormRuleSet$ such that $\alpha$ has been created by an application of $\rho_\alpha$ to $\ChaseEXi{\RuleSet_{NF}, \FactSet}{k}$.

\noindent
It follows from the construction of $\NormRuleSet$ that there must exist:
% 
% \begin{itemize}
% % 
% %  \item  A rule $\Rule_\alpha \in \RuleSet_{II}$ which actually created $\alpha$, of the form 
% %      $\phi(\bar t, \bar y)\wedge M_\gamma \Rightarrow \exists \bar z\; \alpha_0(\bar y,\bar z)$ and  a substitution $\sigma_\alpha$ from the variables in $\bar t\cup \bar y$ to the active domain of $\ChaseEXi{\NormRuleSet, \FactSet}{k}$ such that $\alpha=appl(\Rule_\alpha, \sigma_\alpha)$.
% 

%%%%%%%%%%%%%%%%%5%  
% 
% \end{itemize}
% 
\begin{enumerate}
%     \item  A rule $\Rule_\alpha \in \RuleSet_{II}$ which actually created $\alpha$, of the form 
%     $\phi(\bar t, \bar y)\wedge M_\gamma \Rightarrow \exists \bar z\; \alpha_0(\bar y,\bar z)$ and  a substitution $\sigma_\alpha$ from the variables in $\bar t\cup \bar y$ to the active domain of $\ChaseEXi{\NormRuleSet, \FactSet}{k}$ such that $\alpha=appl(\Rule_\alpha, \sigma_\alpha)$.
% 
% 
%     

    \item  A rule $\Rule'  \in \RuleSet_{I}$ of the form: 
    $$\phi(\bar t, \bar y)\wedge  \gamma(\bar u)\Rightarrow \exists \bar z\; \alpha_0(\bar y,\bar z)$$
     which, by Step II of Normalization Algorithm, led to the creation of $\Rule_\alpha$,
    such that  the tuple $\bar u$ of variables is disjoint  with $\bar y$ and with $\bar t$, and that $\phi(\bar t, \bar y)$ is a connected query.

    \item A rule $\Rule\in \EXRuleSet$ of the form:
   $$
    \beta(\bar x, \bar y)\Rightarrow \exists \bar z\; \alpha_0(\bar y,\bar z)
    $$
    
    which, by Step I of Normalization Algorithm, led to the creation of $\Rule'$,
such that: 
\begin{equation}\tag{$\clubsuit$}
(\exists \bar t,\bar u\; \phi(\bar t, \bar y)\wedge  \gamma(\bar u))\in rew_\ruleset(\exists \bar x \; \beta(\bar x, \bar y) )
\end{equation}
    
    \item  A rule $\Rule_{M_{\gamma}}$ in $\RuleSet_{III}$ of the form $\zeta(\bar u)\Rightarrow M_\gamma$ such that $\exists {{\bar u}}\; \zeta(\bar u)\in rew_\ruleset(\exists {{\bar u}}\; \gamma(\bar u))$ and  a substitution  $\sigma_\zeta$ from the variables in $\bar u$ to  $dom(\ChaseEXi{\NormRuleSet, \FactSet}{k - 1})$ such that $\sigma_\zeta(\zeta(\bar u))\subseteq \ChaseEXi{\NormRuleSet, \FactSet}{k - 1}$. This is because, for $\Rule_\alpha$ to be applicable in  $\ChaseEXi{\NormRuleSet, \FactSet}{k}$ there must be $\ChaseEXi{\NormRuleSet, \FactSet}{k} \models M_\gamma$, and since $\FactSet\not\models M_\gamma$, a rule able to produce $M_\gamma$ must have earlier been applied.
\end{enumerate} 
 
In order to complete the induction step we need to show that $\alpha\in \ChaseEX{\RuleSet, \FactSet}$. Since $\Rule$ is a rule of $\RuleSet$  this claim will follow once we can prove that:
 \begin{equation} \tag{$\diamondsuit$} \label{eq:soundness1}
  \ChaseEX{\RuleSet, \FactSet}\models \exists \bar x \; \beta(\bar x, \sigma_\alpha(\bar y))   
 \end{equation}

%We will prove (\ref{eq:soundness1}) using\footnote{We had quite heated discussion between the authors on whether exercises should be used as parts of the proof itself. However, we agreed on the fact that insights provided here are instrumental in understanding the proof as a whole.} \cref{cw3} if we could show that for some query $\Psi\in rew_\ruleset(\exists \bar x\; \gamma(\bar x, \bar y))$ the following holds:
And (\ref{eq:soundness1}) will follow (using~\cref{cw3}) once we can show that for some query $\varrho(\bar y) \in rew_\ruleset(\exists \bar x\; \beta(\bar x, \bar y))$ there is:
\begin{equation} \tag{$\spadesuit$} \label{eq:soundness2}
\ChaseEX{\RuleSet, \FactSet}\models \varrho(\sigma_\alpha(\bar y))  
\end{equation}
Recall that from $(\heartsuit)$  we know that:
$$\ChaseEXi{\NormRuleSet, \FactSet}{k}\models \exists \bar t \; \phi(\bar t, \sigma_\alpha(\bar y)).$$
By hypothesis this implies that:
$$\ChaseEX{\RuleSet, \FactSet}\models \exists \bar t \; \phi(\bar t, \sigma_\alpha(\bar y)).$$ 
The claim  $(\spadesuit)$
would be now proven, using $(\clubsuit)$, if we could now show that also:
$$\ChaseEX{\RuleSet, \FactSet}\models \exists \bar u\; \gamma(\bar u),$$ 

But this may not be the case. All we know  is that: 
$$\ChaseEXi{\NormRuleSet, \FactSet}{k}\models  \exists \bar{u} \;\zeta(\bar{u})$$
and  therefore (using the hypothesis again) that:  $$\ChaseEX{\RuleSet, \FactSet}\models  \exists \bar{u} \;\zeta(\bar{u}).$$
So we get that $\ChaseEX{\RuleSet, \FactSet}\models  \exists \bar{s}, \bar t \; \phi(\bar t, \sigma_\alpha(\bar y))\wedge \zeta(\bar{s})  $.
Now use \cref{cw10} to get (\ref{eq:soundness2}).
\end{proof}

The proof of the following lemma seems overly complicated. Why doesn't it  just follow
 from \cref{cw7}? 
This is because, while we assume that $\ruleset$ is BDD, we never proved that  $\NormRuleSet$ is BDD too. 

\begin{lemma}\label{lem:nullary-one-step}\label{lem:skeleton-C}
Suppose  $\Chase{\NormRuleSet, \FactSet}\models M_\phi$ for some nullary fact $M_\phi$.
Then  $Ch_1(\NormRuleSet, \FactSet)\models M_\phi$.
\end{lemma}
\begin{proof}
Suppose  $\Chase{\NormRuleSet, \FactSet}\models M_\phi$ 
 and let  $\beta({\bar x})$ be the body of the rule from $\ruleset_{III}$ that created $M_\phi$.

Recall that $\beta({\bar x})$ does not contain nullary predicates and 
that $\NormRuleSet$ consists only of existential rules (from $\ruleset_{II}$) and of datalog rules that prove nullary facts 
(from $\ruleset_{III}$). This means that $\ChaseEX{\NormRuleSet, \FactSet} \models \exists{{\bar x}} \; \beta({\bar x})$.

Now we can use \cref{lem:AB} and  get: $\ChaseEX{\RuleSet, \FactSet} \models \exists{{\bar x}}\;\beta({\bar x})$.

Now recall that $\exists{{\bar x}}\;\beta({\bar x})\in rew_\ruleset(\phi)$. From \cref{cw7} we get that there exists another boolean query $\beta_0$ 
such that $\beta_0\in  rew_\ruleset(\phi)$ and that $\FactSet \models \beta_0$. 
And from this we can immediately conclude that $\Step{1}{\NormRuleSet, \FactSet} \models M_\phi$.
\end{proof}

Combining \cref{lem:skeleton-C} and \cref{obs:detachedrulenullbody} we get:
\begin{corollary}\label{lem:skeleton-Ca}
If  $\Chase{\NormRuleSet, \FactSet}\models \alpha$ for some detached atom $\alpha$  then 
$Ch_2(\NormRuleSet, \FactSet)\models \alpha$.

\end{corollary}

Last lemma we need for the proof of Lemma \ref{lem:skeleton_equivalence} is:

\begin{lemma}\label{lem:skeleton-D}
$\ChaseEXi{\RuleSet, \FactSet}{i} \subseteq \ChaseEXi{\NormRuleSet, \FactSet}{i + 2}$
\end{lemma}
\begin{proof}
The base of induction $\FactSet \subseteq \ChaseEXi{\NormRuleSet, \FactSet}{2}$ is clearly true. And so let us move to the induction step.

Let $\alpha$ be an atom produced by an existential rule during the $i+1$-th step of $\Chase{\RuleSet, \FactSet}$. Let rule $\rho \in \RuleSet$ and mapping $\sigma$ be such that $\alpha = appl(\rho, \sigma)$ and let $\gamma({\bar x}, {\bar y})$ be a body of $\rho$ where ${\bar y}$ are the frontier variables.

While we know that $\Step{i}{\RuleSet, \FactSet} \models \sigma(\gamma({\bar x}, {\bar y}))$ we are not sure that $\ChaseEXi{\RuleSet, \FactSet}{i} \models \sigma(\gamma({\bar x}, {\bar y}))$: some of the atoms in $\sigma(\gamma({\bar x}, {\bar y}))$ could be produced by datalog rules of $\ruleset$.

But $Ch(\RuleSet, \FactSet)=Ch(\RuleSet,\ChaseEXi{\RuleSet, \FactSet}{i})$, so if $\Step{i}{\RuleSet, \FactSet} \models \sigma(\gamma({\bar x}, {\bar y}))$ then we can be sure that there exists a query $\exists \bar z \beta(\bar z,\bar y)\in rew_\ruleset(\exists \bar x \gamma(\bar x,\bar y))$ such that $\ChaseEXi{\RuleSet, \FactSet}{i} \models \exists \bar z \beta(\bar z,\bar \sigma(y))$ or,
in other words,  $\ChaseEXi{\RuleSet, \FactSet}{i} \models  \beta(\sigma_\beta(\bar z),\bar \sigma(\bar y))$ for some
substitution $\sigma_\beta$.

 Observe that from induction hypothesis we know that $\ChaseEXi{\NormRuleSet, \FactSet}{i + 2} \models 
 \beta(\sigma_\beta(\bar z),\bar \sigma(\bar y))$.

 Let now $\rho'$ be the rule from  $\NormRuleSet$ (or, to be more precise, from $\ruleset_{II}$), such that 
 $head(\rho')=head(\rho)$ and that the body of $\rho'$ is of the form 
 $\beta_0(\bar u,\bar y)\wedge M_\phi$, where   $ \beta(\bar z,\bar y) = \beta_0(\bar u,\bar y)\wedge \phi(\bar v)$ for some 
 $\bar v$ disjoint from $\bar u\cup \bar y$. It follows from the construction of $\NormRuleSet$ that such $\rho'$ exits.

Let us consider the second case. Note that $\ChaseEXi{\NormRuleSet, \FactSet}{i + 2} \models \sigma_\beta(\beta'({\bar x}, {\bar y}))$ and so all we need to show in order to complete our induction step  is that $\Step{i + 2}{\NormRuleSet, \FactSet} \models M_\phi$.

Clearly $\ChaseEXi{\NormRuleSet, \FactSet}{i  + 2} \models \exists {{\bar z}}\phi({\bar z})$ as we know that $\phi \subseteq \beta$. From this we can conclude that $\FactSet \models rew_\ruleset(\phi)$ and thus $M_\phi \in \ChaseEXi{\NormRuleSet, \FactSet}{2}$. 

This ends the proof of Lemma \ref{lem:skeleton-D} and of Lemma \ref{lem:skeleton_equivalence}.
\end{proof}

Note that this implies that sensible and non-nullary atoms of $\Chase{\NormRuleSet, \FactSet}$ form a set of trees having a tree $\Skeleton(t)$ for each term $t$ that is a constant of $\FactSet$ or detached term of $\ChaseEX{\NormRuleSet, \FactSet}$. And that this set of trees is exactly this same set of trees as in the case of rule set $\RuleSet$. Also, note that having \cref{lem:skeleton_equivalence} we get a very important: 

\begin{corollary} \label{lem:dl_over_skeleton}
For any set of facts $\FactSet$:
$$ \Chase{\DLRuleSet, \ChaseEX{\NormRuleSet, \FactSet} \cup \FactSet} = \Chase{\RuleSet, \FactSet} $$\\
\end{corollary}

\subsection{The Crucial Lemma}
In this section we state and prove the Crucial Lemma. Recall, that we want to prove, 
that for any term $t$ being constant of $\FactSet$ or detached term of $\ChaseEX{\NormRuleSet, \FactSet}$ the tree $\Skeleton(t)$ rooted in $t$ and consisting of sensible atoms of $\ChaseEX{\NormRuleSet, \FactSet}$ requires only a small subset of $\FactSet$ to be built by the chase $\Chase{\NormRuleSet, \FactSet}$.

First let us distinguish, among the parents and ancestors of some atom in
$\Chase{\NormRuleSet, \FactSet}$, its {\em connected parents} and {\em connected ancestors}:
given some parent function $\ParentFunction_{\NormRuleSet}$ for $\Chase{\NormRuleSet, \FactSet}$ and an atom $\alpha$ of $\Chase{\NormRuleSet, \FactSet}$ we define the set connected parents of $\cparent(\alpha)$ as the set of 
all the non-nullary atoms of
$\ParentFunction_{\NormRuleSet}(\alpha)$. Then for atoms of $\Chase{\NormRuleSet, \FactSet}$ we define their respective sets of connected ancestors as follows:
\begin{itemize}
    \item $\cancestor(\alpha) = \set{\alpha}$ for an atom $\alpha \in \FactSet$,
    \item $\cancestor(\alpha) = \bigcup_{\alpha' \in \cparent(\alpha)} \cancestor(\alpha')$.
\end{itemize}

Before going onward, let us define a number of constants, depending on $\NormRuleSet$. \begin{itemize}
\item Let $k$ be the number of nullary predicates in $\NormRuleSet$.
\item Let $h$ be the maximal number of atoms in the body of a rule from $\NormRuleSet$.
\item Let $n$ be the number of rules in $\NormRuleSet$.
\item Let $N$ be the number of elements of a full $n$-ary tree of depth $h$.
\item Let $M = N h + k h$
\end{itemize}

Now having all prepared we are going to prove the Crucial Lemma:

\begin{lemma}[crucial]\label{lem:crucial}
For every set of facts $\FactSet$, every term $t$ that is constant of $\FactSet$ or detached term of $\ChaseEX{\NormRuleSet, \FactSet}$ and every ancestor function $\AncestorFunction_{\NormRuleSet}$:
$$\left\vert \bigcup\nolimits_{\alpha \in \Skeleton(t)} \AncestorFunction_{\NormRuleSet} (\alpha) \right\vert \leq  M$$
\end{lemma}

Notice that we now count ancestors with respect to $\NormRuleSet$ instead of 
ancestors with respect to $\RuleSet$. This is, in fact, the only difference between this lemma and 
the (false) Lemma \ref{false-crucial-lemma}.

\begin{proof}
The atoms of $\Skeleton(t)$ are all produced by the  rules of $\RuleSet_{II}$. Recall that bodies of those  rules consist of one connected $CQ$ and  one nullary atom. The number of atoms that are ancestors of nullary atoms can easily be bound by $k h$ using \cref{lem:nullary-one-step}. 

What is left to be bound is the number of connected ancestors of $\Skeleton(t)$.
Recall that for $t\in dom(\factset)\cup det(\Chase{\NormRuleSet, \FactSet})$ the set of facts
$\Skeleton(t)$ is a tree with $t$ as its root. This gives us a natural notion of depth of atoms in $\Skeleton(t)$, with atoms containing $t$ being at depth one. 

First we will consider the easier case when term $t\in det(\Chase{\RuleSet, \FactSet})$. 
Then $S(t)$ is a detached tree, not connected, by atoms of $\Chase{\NormRuleSet, \FactSet})$ to $\factset$. 
And  the connected parents, and connected ancestors are defined in such a way that every atom is in one connected component of 
$\Chase{\NormRuleSet, \FactSet})$  with all its connected ancestors. So, no atom in  $S(t)$ has any connected ancestors,
and thus the entire  $S(t)$ has in total at most $k h$ ancestors.

Let us now move the case when $t\in dom(\FactSet)$.
Clearly, to be a connected ancestor of someone in $\Skeleton(t)$ an atom in \factset must be a connected parent of someone in $\Skeleton(t)$:

$$\bigcup\nolimits_{\alpha \in \Skeleton(t)} \cancestor(\alpha) = \left(\bigcup\nolimits_{\alpha \in \Skeleton(t)} \cparent(\alpha) \right)  \cap \FactSet,$$

using the above equality one can bound the number of ancestors of atoms in connected ancestor of someone in $\Skeleton(t)$ in the following way:
$$\left|\left(\bigcup\nolimits_{\alpha \in \Skeleton(t)} \cparent(\alpha) \right)  \cap \FactSet\right| \leq \sum_{\alpha\in\Skeleton(t)} \left| \cparent(\alpha) \cap \FactSet \right|.$$

Now, we can easily bound $\left|\cparent(\alpha) \cap \FactSet\right|$ with $h$ for any atom $\alpha$ of $\Chase{\NormRuleSet, \FactSet}$. Finally we are going to show that only a finite (and bounded)  number of atoms of $\Skeleton(t)$ have any connected parents in $\FactSet$. Observe that, if an application of an existential rule $\rho$ created an atom that is at depth greater than $h$   then the connected part of the  body $\rho$, of size at most $h$, could not ``reach'' the atoms of $\FactSet$. From this we get $Nh$ bound for the number of connected ancestors of $\Skeleton(t)$ which gives us result of $M$ being the bound on number of ancestors of $\Skeleton(t)$ for any $t$.
\end{proof}

\subsection{Proving Theorem \ref{th-binary-local}}

Now we can finally prove Theorem \ref{th-binary-local}.
The following is an easy corollary of \cref{lem:crucial} and \cref{lem:dl_over_skeleton}:

\begin{corollary}\label{cor:crucialforruleset}
For any term $t\in dom(\factset)\cup det(\Chase{\NormRuleSet, \FactSet})$  there exists a subset $\FactSet'$ of $\FactSet$ of size at most $M$ such that $\Skeleton(t) \subseteq \Chase{\RuleSet, \FactSet'}$.
\end{corollary}

This is almost  Theorem \ref{th-binary-local}, but only for existential atoms of $\Chase{\RuleSet, \FactSet}$. However, observe that:
\begin{observation}\label{obs:dl_ancestors}
There exists a constant $d_\RuleSet$ such that for any set of facts $\FactSet$ and for any atom $\alpha$ of $\Chase{\DLRuleSet, \FactSet}$ there exists a subset $\FactSet'$ of $\FactSet$ such that $|\FactSet'| < d_\RuleSet$ and $\alpha \in \Chase{\RuleSet, \FactSet'}$.
\end{observation}
\begin{proof}
From~\cref{soon} the size of $\FactSet'$ can easily be bound by $h^{n_{at}}$ where $h$ is the maximal number of atoms in a rule of $\RuleSet$ and $n_{at}$ is a constant from~\cref{soon}.
\end{proof}

Thus $Md_{\RuleSet}$ is a locality constant for Datalog atoms of $\Chase{\RuleSet, \FactSet}$ this concludes the proof of~\cref{th-binary-local}.

\newpage

\section{Appendix B: Proof of Lemma \ref{soundness-lemma} }\label{soundness-lemma-proof}\label{app-2}

Let $Q=\langle \phi(\bar y) , V  \rangle$ be any live  marked query.

Lemma \ref{soundness-lemma} will follow directly from Lemmas~\ref{lem:dav-sound-cut-green}--\ref{lem:dav-sound-reduce}:

\begin{lemma}\label{lem:dav-sound-cut-green}
For any set of facts $\FactSet$, tuple  $\aaaddd{\factset}$ and $x$ such that $\M' = \cutgreen(\M, x)$ 
(or $\M' = \cutred(\M, x)$) the following holds:
$$\Chase{\RuleSetD, \FactSet} \models \M({\bar a}) \iff \Chase{\RuleSetD, \FactSet} \models \M'({\bar a}).$$
\end{lemma}
\begin{proof} 
Let $G(z,x)$ be the atom removed  from $q(Q)$ by $\cutgreen$.

$(\Rightarrow)$. It follows immediately as $q(\M')$ is a subset of $q(\M)$ and $V(\M') = V(\M)$.

$(\Leftarrow)$. Let $h'$ be a homomorphism witnessing that $\Chase{\RuleSetD, \FactSet} \models \M'({\bar a})$. 
Due to the rule $(pins)$ of $ruleset_d$ there must exist $t\in dom(\Chase{\RuleSetD, \FactSet})$ such that 
$G(h'(z),a)\in \Chase{\RuleSetD, \FactSet}$. Define a new homomorphism $h$ as $h'\cup \pair{x, t}$. Then $h$ is
a homomorphism witnessing that $\Chase{\RuleSetD, \FactSet} \models \M({\bar a})$. 

The proof for $\cutred$ is analogous.
\end{proof}

\begin{lemma}\label{lem:dav-sound-fuse-green}
For any set of facts $\FactSet$, tuple  $\aaaddd{\factset}$  and $x, z, z'$ such that $\M' = \mergegreen(\M, x, z, z')$
(or $\M' = \mergered(\M, x, z, z')$)    the following holds:
$$\Chase{\RuleSetD, \FactSet} \models \M({\bar a}) \iff \Chase{\RuleSetD, \FactSet} \models \M'({\bar a}).$$
\end{lemma}
\begin{proof}
Note that there exists a homomorphism (the one that identifies $z'$ with $z$) from $q(\M)$ to $q(\M')$ which preserves markings of variables. Thus $(\Leftarrow)$ is trivial.

$(\Rightarrow)$. Let $h$ be a homomorphism witnessing that $\Chase{\RuleSetD, \FactSet} \models \M({\bar a})$. 
If we can show that 
 $h(z) = h(z')$ then $h$ will witness that $\Chase{\RuleSetD, \FactSet} \models \M'({\bar a})$ as well. 
 First recall that, since $x\not\in V(Q)$ we can be sure that $h(x)\in dom(\Chase{\RuleSetD, \FactSet} \setminus dom(\factset)$.  
 Then notice that it follows from the rules of $\ruleset_d $ that the  in-degree, with respect to relation $G$,
 of any term of $dom(\Chase{\RuleSetD, \FactSet} \setminus dom(\factset)$  is at most one. Thus the only way for
 $\Chase{\RuleSetD, \FactSet} \models \M({\bar a})$ to happen is that $h(z) = h(z')$ thus $\Chase{\RuleSetD, \FactSet} \models \M'({\bar a})$.
The proof for  $\mergered$ is analogous.
\end{proof}

\begin{lemma}\label{lem:dav-sound-reduce}
For any set of facts $\FactSet$, tuple  $\aaaddd{\factset}$
and $x$ such that $\setM = \reduce(\M, x)$ the following holds:
$$\Chase{\RuleSetD, \FactSet} \models \M({\bar a}) \iff \exists_{\M' \in \setM} \Chase{\RuleSetD, \FactSet} \models \M'({\bar a}).$$
\end{lemma}
\begin{proof}
Let $x'$, $x''$, $x_r$, and $x_g$ be variables of $q(\M)$ such that $R(x_r, x), G(x_g, x) \in q(\M)$ and that $R(x',x_g), G(x', x''), G(x'', x_r) \in q(\M')$.

$(\Leftarrow)$. Let $\M' \in \setM$ be such that $\Chase{\RuleSetD, \FactSet} \models \M'({\bar a})$ and let $h'$ be a homomorphism witnessing that. We need to
show that there exists a homomorphism $h$ witnessing $\Chase{\RuleSetD, \FactSet} \models \M({\bar a})$.

Since $R(x',x_g), G(x', x''), G(x'', x_r)$ are atoms of $q(Q)$, and $h'$ is a homomorphism, we know that 
$R(h'(x_g), h'(x')), G(h'(x'), h'(x'')), G(h'(x''), h'(x_r))$ are atoms of $\Chase{\RuleSetD, \FactSet}$. But, since $(grid)$ is a rule of $\ruleset_d$,
this implies that there exists an element $t\in \Chase{\RuleSetD, \FactSet}$, such that $G(h'(x_g), h'(t)) $ and  $R(h'(x_r), h'(t)) $
are also in  $\Chase{\RuleSetD, \FactSet}$.

Define $h$ as:
\begin{itemize}
    \item $h(u) = h'(u)$ for $u \in var(\M) \setminus \set{x}$.
    \item $h(x) = t$.
\end{itemize}

% +++
$(\Rightarrow)$. Let $h$ be a homomorphism witnessing that $\Chase{\RuleSetD, \FactSet} \models \M({\bar a})$. We will show that there exist $\M' \in \setM$ and a homomorphism $h'$ such that $h'$ is witnessing $\Chase{\RuleSetD, \FactSet} \models \M'({\bar a})$. Again recall that $x$ is an unmarked variable and so let us find parents of $h(x)$ in $\Chase{\RuleSetD, \FactSet}$. Note that $h(x)$ could be created only by rule $(grid)$ as it has an in-degree of two. Let $\sigma$ be such that $appl((grid), \sigma) = R(h(x_r), h(x)), G(h(x_g), h(x))$.
We set $h' = h \setminus \set{x, h(x)} \cup \set{\pair{x', \sigma(x)}, \pair{x'', \sigma(y)}}$.
%We set $h' = h \setminus \set{x, h(x)} \cup \set{\pair{x', \sigma(x')}, \pair{x'', \sigma(x'')}}$.
There are four possible picks for $\M'$ from $\setM$. While homomorphism $h'$ works for any of them as elements of $\setM$ differ only by markings, we need to make sure that the marking of $\M'$ agrees with $h'$. Obviously for every variable of $\M$ this is the case. However we have two new variables to consider namely $x'$ and $x''$. Thus we need to take $\M'$ such that it satisfies: 
\begin{itemize}
    \item $h'(x')$ is a constant of $\FactSet$ if and only if $x' \in V(\M')$.
    \item $h'(x'')$ is a constant of $\FactSet$ if and only if $x'' \in V(\M')$.
\end{itemize} 
This is trivially possible from~\cref{reduce-def}.
\end{proof}

\end{document}